\documentclass[12pt,a4paper]{article}
\usepackage{graphicx} 
\usepackage{lmodern} 
\usepackage{mathtools}
\usepackage{amsthm}
\usepackage{authblk} 
\usepackage[colorlinks=true,linkcolor=magenta]{hyperref}
\usepackage{physics}
\usepackage{ascmac}
\usepackage{here}
\usepackage[most]{tcolorbox}
\usepackage{mathtools}
\usepackage{bm}
\usepackage{cite}             
\usepackage{enumitem}    
\usepackage{tikz-cd}

\usepackage[english]{babel}
\usepackage{color}
\usepackage{comment}
\usepackage[margin=1in]{geometry}
\usepackage{color}

\usepackage{ulem}
\usepackage{amsmath,amssymb,amsfonts,mathrsfs}

\usepackage[T1]{fontenc}
\usepackage{cleveref}

\usepackage{graphicx}
\newcommand{\mr}[1]{\mathrm{#1}}

\newcommand{\mc}[1]{\mathcal{#1}}

\newcommand{\mb}[1]{\mathbb{#1}}

\newcommand{\mf}[1]{\mathfrak{#1}}
\newcommand{\what}[1]{\widehat{#1}}
\newcommand{\wtil}[1]{\widetilde{#1}}

\newtheorem{theorem}{Theorem}[section]
\newtheorem{lemma}[theorem]{Lemma}
\newtheorem{corollary}[theorem]{Corollary}
\newtheorem{prop}[theorem]{Proposition}
\theoremstyle{definition}

\theoremstyle{remark}
\newtheorem*{remark}{Remark}

\title{Any local Hamiltonian with ferromagnetic quantum many-body scars has a generalized Shiraishi-Mori form}
\author{Keita Omiya}
\affil{Institute of Physics, Ecole Polytechnique F\'ed\'erale de Lausanne (EPFL), CH-1015 Lausanne, Switzerland}
\date{\today}

\begin{document}

\maketitle
\begin{abstract}
    Quantum many-body scars (QMBS) are nonthermal eigenstates embedded in otherwise thermal spectra. A broad class of exact QMBS is realized as fixed-momentum magnon states above a ferromagnetic reference state. Here we prove a structural theorem for this class. Specifically, we show that any local Hamiltonian hosting such ``ferromagnetic scar states'' necessarily admits a decomposition into a Zeeman term and terms containing local projectors that annihilate the scar states locally. This result establishes that an appropriate generalization of the Shiraishi--Mori construction\,\cite{mori_shiraishi} is essentially exhaustive for ferromagnetic QMBS and provides a unified structural explanation for the recurrent appearance of projector-based interactions and equally spaced scar towers across a broad family of exact scar Hamiltonians.
\end{abstract}
\section{Introduction}

\label{sec:intro}
Quantum many-body scars (QMBS) constitute a striking form of weak ergodicity breaking\,\cite{Serbyn2021,Moudgalya2022,Chandran2023}. Unlike other known non-thermal systems\,\cite{Baxter1982,PhysRevLett.111.127201,PhysRevB.90.174202,mbl_review,Nandkishore2015}, non-thermal behavior of QMBS does not originate from local conserved quantities\footnote{Note that some studies have cast doubt on the existence of quasi-local integrals of motion (LIOMs) in many-body localization\,\cite{untajs2020,Sels_2021}.}, but rather from a sparse set of atypical eigenstates, often called scar states, that strongly violate eigenstate thermalization hypothesis (ETH)\,\cite{Deutsch1991,Srednicki1994,eth_review}. 
Modern interest in such states was sparked by experiments on Rydberg-atom arrays\,\cite{Bernien2017} and by their effective description in terms of the PXP model\,\cite{Turner2018,Turner2018prb}. 
Motivated by these studies, numerous toy models that host exact scar eigenstates have been constructed, as well as general construction principles for such Hamiltonians\,\cite{mori_shiraishi,tos_xy,tos_aklt_unified,topological_scar,Zhao2020,Hudomal2020,Kiryl2020,Pakrouski2021,ODea2020TunnelsToTowers,Ren2021Quasisymmetry}. The list of such Hamiltonians has continued to expand, including standard condensed matter models such as the Hubbard model\,\cite{PhysRevLett.63.2144,PhysRevB.102.075132,PhysRevB.102.085140} and the Affleck-Kennedy-Lieb-Tasaki (AKLT) model\,\cite{AROVAS1989431,aklt_nonthermal,tos_aklt_unified}. Remarkably, across many seemingly unrelated examples, their scar states often take a closely parallel form: they are realized as ``magnon'' states obtained by repeatedly applying a simple ladder-like operator of fixed momentum (typically $k=\pi$) to a reference product state. Unlike conventional magnons, these states are not restricted to the low-energy sector, yet they remain exact eigenstates. In this paper, we refer to such states as \textit{ferromagnetic scar states}. 

Equally striking is a recurring structural motif in the corresponding Hamiltonians. In many models hosting ferromagnetic scars, the Hamiltonian can be written as the sum of an ``annihilator'' and a Zeeman term: the scar states are annihilated by the former and are therefore insensitive to the interaction part, while the latter produces an equidistant energy spectrum within the scar manifold. Dynamically, initial states supported on the scar subspace exhibit coherent revivals that can be interpreted as collective large-spin precession. Moreover, the annihilator itself can often be further decomposed into a sum of terms built from local projectors that annihilate the scar states locally. This pattern has been identified not only in a wide range of QMBS models (even when scar states are not ferromagnetic), but also in some models, including the AKLT chain, that were initially thought to evade such a characterization\,\cite{PhysRevA.107.023318,PhysRevB.108.054412}.   

\begin{figure}
    \centering
\tikzset{every picture/.style={line width=0.75pt}} 
\begin{tikzpicture}[x=0.75pt,y=0.75pt,yscale=-0.7,xscale=0.7]

\draw    (110,60) -- (132.15,29.87) ;
\draw [shift={(133.33,28.26)}, rotate = 126.32] [color={rgb, 255:red, 0; green, 0; blue, 0 }  ][line width=0.75]    (10.93,-3.29) .. controls (6.95,-1.4) and (3.31,-0.3) .. (0,0) .. controls (3.31,0.3) and (6.95,1.4) .. (10.93,3.29)   ;
\draw    (130,60) -- (152.15,29.87) ;
\draw [shift={(153.33,28.26)}, rotate = 126.32] [color={rgb, 255:red, 0; green, 0; blue, 0 }  ][line width=0.75]    (10.93,-3.29) .. controls (6.95,-1.4) and (3.31,-0.3) .. (0,0) .. controls (3.31,0.3) and (6.95,1.4) .. (10.93,3.29)   ;
\draw    (150,60) -- (172.15,29.87) ;
\draw [shift={(173.33,28.26)}, rotate = 126.32] [color={rgb, 255:red, 0; green, 0; blue, 0 }  ][line width=0.75]    (10.93,-3.29) .. controls (6.95,-1.4) and (3.31,-0.3) .. (0,0) .. controls (3.31,0.3) and (6.95,1.4) .. (10.93,3.29)   ;
\draw    (170,60) -- (192.15,29.87) ;
\draw [shift={(193.33,28.26)}, rotate = 126.32] [color={rgb, 255:red, 0; green, 0; blue, 0 }  ][line width=0.75]    (10.93,-3.29) .. controls (6.95,-1.4) and (3.31,-0.3) .. (0,0) .. controls (3.31,0.3) and (6.95,1.4) .. (10.93,3.29)   ;
\draw    (190,60) -- (212.15,29.87) ;
\draw [shift={(213.33,28.26)}, rotate = 126.32] [color={rgb, 255:red, 0; green, 0; blue, 0 }  ][line width=0.75]    (10.93,-3.29) .. controls (6.95,-1.4) and (3.31,-0.3) .. (0,0) .. controls (3.31,0.3) and (6.95,1.4) .. (10.93,3.29)   ;
\draw    (210,60) -- (232.15,29.87) ;
\draw [shift={(233.33,28.26)}, rotate = 126.32] [color={rgb, 255:red, 0; green, 0; blue, 0 }  ][line width=0.75]    (10.93,-3.29) .. controls (6.95,-1.4) and (3.31,-0.3) .. (0,0) .. controls (3.31,0.3) and (6.95,1.4) .. (10.93,3.29)   ;
\draw    (230,60) -- (252.15,29.87) ;
\draw [shift={(253.33,28.26)}, rotate = 126.32] [color={rgb, 255:red, 0; green, 0; blue, 0 }  ][line width=0.75]    (10.93,-3.29) .. controls (6.95,-1.4) and (3.31,-0.3) .. (0,0) .. controls (3.31,0.3) and (6.95,1.4) .. (10.93,3.29)   ;

\draw    (331,59) -- (353.15,28.87) ;
\draw [shift={(354.33,27.26)}, rotate = 126.32] [color={rgb, 255:red, 0; green, 0; blue, 0 }  ][line width=0.75]    (10.93,-3.29) .. controls (6.95,-1.4) and (3.31,-0.3) .. (0,0) .. controls (3.31,0.3) and (6.95,1.4) .. (10.93,3.29)   ;
\draw    (351,59) -- (373.15,28.87) ;
\draw [shift={(374.33,27.26)}, rotate = 126.32] [color={rgb, 255:red, 0; green, 0; blue, 0 }  ][line width=0.75]    (10.93,-3.29) .. controls (6.95,-1.4) and (3.31,-0.3) .. (0,0) .. controls (3.31,0.3) and (6.95,1.4) .. (10.93,3.29)   ;
\draw    (371,59) -- (393.15,28.87) ;
\draw [shift={(394.33,27.26)}, rotate = 126.32] [color={rgb, 255:red, 0; green, 0; blue, 0 }  ][line width=0.75]    (10.93,-3.29) .. controls (6.95,-1.4) and (3.31,-0.3) .. (0,0) .. controls (3.31,0.3) and (6.95,1.4) .. (10.93,3.29)   ;
\draw    (391,59) -- (413.15,28.87) ;
\draw [shift={(414.33,27.26)}, rotate = 126.32] [color={rgb, 255:red, 0; green, 0; blue, 0 }  ][line width=0.75]    (10.93,-3.29) .. controls (6.95,-1.4) and (3.31,-0.3) .. (0,0) .. controls (3.31,0.3) and (6.95,1.4) .. (10.93,3.29)   ;
\draw    (411,59) -- (433.15,28.87) ;
\draw [shift={(434.33,27.26)}, rotate = 126.32] [color={rgb, 255:red, 0; green, 0; blue, 0 }  ][line width=0.75]    (10.93,-3.29) .. controls (6.95,-1.4) and (3.31,-0.3) .. (0,0) .. controls (3.31,0.3) and (6.95,1.4) .. (10.93,3.29)   ;
\draw    (431,59) -- (453.15,28.87) ;
\draw [shift={(454.33,27.26)}, rotate = 126.32] [color={rgb, 255:red, 0; green, 0; blue, 0 }  ][line width=0.75]    (10.93,-3.29) .. controls (6.95,-1.4) and (3.31,-0.3) .. (0,0) .. controls (3.31,0.3) and (6.95,1.4) .. (10.93,3.29)   ;
\draw    (451,59) -- (473.15,28.87) ;
\draw [shift={(474.33,27.26)}, rotate = 126.32] [color={rgb, 255:red, 0; green, 0; blue, 0 }  ][line width=0.75]    (10.93,-3.29) .. controls (6.95,-1.4) and (3.31,-0.3) .. (0,0) .. controls (3.31,0.3) and (6.95,1.4) .. (10.93,3.29)   ;

\draw   (64.33,28.26) .. controls (64.33,22.74) and (68.81,18.26) .. (74.33,18.26) -- (478.33,18.26) .. controls (483.86,18.26) and (488.33,22.74) .. (488.33,28.26) -- (488.33,58.26) .. controls (488.33,63.78) and (483.86,68.26) .. (478.33,68.26) -- (74.33,68.26) .. controls (68.81,68.26) and (64.33,63.78) .. (64.33,58.26) -- cycle ;
\draw    (44,203) -- (66.15,172.87) ;
\draw [shift={(67.33,171.26)}, rotate = 126.32] [color={rgb, 255:red, 0; green, 0; blue, 0 }  ][line width=0.75]    (10.93,-3.29) .. controls (6.95,-1.4) and (3.31,-0.3) .. (0,0) .. controls (3.31,0.3) and (6.95,1.4) .. (10.93,3.29)   ;
\draw    (64,203) -- (86.15,172.87) ;
\draw [shift={(87.33,171.26)}, rotate = 126.32] [color={rgb, 255:red, 0; green, 0; blue, 0 }  ][line width=0.75]    (10.93,-3.29) .. controls (6.95,-1.4) and (3.31,-0.3) .. (0,0) .. controls (3.31,0.3) and (6.95,1.4) .. (10.93,3.29)   ;
\draw    (84,203) -- (106.15,172.87) ;
\draw [shift={(107.33,171.26)}, rotate = 126.32] [color={rgb, 255:red, 0; green, 0; blue, 0 }  ][line width=0.75]    (10.93,-3.29) .. controls (6.95,-1.4) and (3.31,-0.3) .. (0,0) .. controls (3.31,0.3) and (6.95,1.4) .. (10.93,3.29)   ;
\draw    (104,203) -- (126.15,172.87) ;
\draw [shift={(127.33,171.26)}, rotate = 126.32] [color={rgb, 255:red, 0; green, 0; blue, 0 }  ][line width=0.75]    (10.93,-3.29) .. controls (6.95,-1.4) and (3.31,-0.3) .. (0,0) .. controls (3.31,0.3) and (6.95,1.4) .. (10.93,3.29)   ;
\draw    (124,203) -- (146.15,172.87) ;
\draw [shift={(147.33,171.26)}, rotate = 126.32] [color={rgb, 255:red, 0; green, 0; blue, 0 }  ][line width=0.75]    (10.93,-3.29) .. controls (6.95,-1.4) and (3.31,-0.3) .. (0,0) .. controls (3.31,0.3) and (6.95,1.4) .. (10.93,3.29)   ;
\draw    (144,203) -- (166.15,172.87) ;
\draw [shift={(167.33,171.26)}, rotate = 126.32] [color={rgb, 255:red, 0; green, 0; blue, 0 }  ][line width=0.75]    (10.93,-3.29) .. controls (6.95,-1.4) and (3.31,-0.3) .. (0,0) .. controls (3.31,0.3) and (6.95,1.4) .. (10.93,3.29)   ;
\draw    (164,203) -- (186.15,172.87) ;
\draw [shift={(187.33,171.26)}, rotate = 126.32] [color={rgb, 255:red, 0; green, 0; blue, 0 }  ][line width=0.75]    (10.93,-3.29) .. controls (6.95,-1.4) and (3.31,-0.3) .. (0,0) .. controls (3.31,0.3) and (6.95,1.4) .. (10.93,3.29)   ;

\draw    (306,202) -- (328.15,171.87) ;
\draw [shift={(329.33,170.26)}, rotate = 126.32] [color={rgb, 255:red, 0; green, 0; blue, 0 }  ][line width=0.75]    (10.93,-3.29) .. controls (6.95,-1.4) and (3.31,-0.3) .. (0,0) .. controls (3.31,0.3) and (6.95,1.4) .. (10.93,3.29)   ;
\draw    (326,202) -- (348.15,171.87) ;
\draw [shift={(349.33,170.26)}, rotate = 126.32] [color={rgb, 255:red, 0; green, 0; blue, 0 }  ][line width=0.75]    (10.93,-3.29) .. controls (6.95,-1.4) and (3.31,-0.3) .. (0,0) .. controls (3.31,0.3) and (6.95,1.4) .. (10.93,3.29)   ;
\draw    (346,202) -- (368.15,171.87) ;
\draw [shift={(369.33,170.26)}, rotate = 126.32] [color={rgb, 255:red, 0; green, 0; blue, 0 }  ][line width=0.75]    (10.93,-3.29) .. controls (6.95,-1.4) and (3.31,-0.3) .. (0,0) .. controls (3.31,0.3) and (6.95,1.4) .. (10.93,3.29)   ;
\draw    (366,202) -- (388.15,171.87) ;
\draw [shift={(389.33,170.26)}, rotate = 126.32] [color={rgb, 255:red, 0; green, 0; blue, 0 }  ][line width=0.75]    (10.93,-3.29) .. controls (6.95,-1.4) and (3.31,-0.3) .. (0,0) .. controls (3.31,0.3) and (6.95,1.4) .. (10.93,3.29)   ;
\draw    (386,202) -- (408.15,171.87) ;
\draw [shift={(409.33,170.26)}, rotate = 126.32] [color={rgb, 255:red, 0; green, 0; blue, 0 }  ][line width=0.75]    (10.93,-3.29) .. controls (6.95,-1.4) and (3.31,-0.3) .. (0,0) .. controls (3.31,0.3) and (6.95,1.4) .. (10.93,3.29)   ;
\draw    (406,202) -- (428.15,171.87) ;
\draw [shift={(429.33,170.26)}, rotate = 126.32] [color={rgb, 255:red, 0; green, 0; blue, 0 }  ][line width=0.75]    (10.93,-3.29) .. controls (6.95,-1.4) and (3.31,-0.3) .. (0,0) .. controls (3.31,0.3) and (6.95,1.4) .. (10.93,3.29)   ;
\draw    (426,202) -- (448.15,171.87) ;
\draw [shift={(449.33,170.26)}, rotate = 126.32] [color={rgb, 255:red, 0; green, 0; blue, 0 }  ][line width=0.75]    (10.93,-3.29) .. controls (6.95,-1.4) and (3.31,-0.3) .. (0,0) .. controls (3.31,0.3) and (6.95,1.4) .. (10.93,3.29)   ;

\draw    (507,204) -- (529.15,173.87) ;
\draw [shift={(530.33,172.26)}, rotate = 126.32] [color={rgb, 255:red, 0; green, 0; blue, 0 }  ][line width=0.75]    (10.93,-3.29) .. controls (6.95,-1.4) and (3.31,-0.3) .. (0,0) .. controls (3.31,0.3) and (6.95,1.4) .. (10.93,3.29)   ;
\draw    (527,204) -- (549.15,173.87) ;
\draw [shift={(550.33,172.26)}, rotate = 126.32] [color={rgb, 255:red, 0; green, 0; blue, 0 }  ][line width=0.75]    (10.93,-3.29) .. controls (6.95,-1.4) and (3.31,-0.3) .. (0,0) .. controls (3.31,0.3) and (6.95,1.4) .. (10.93,3.29)   ;
\draw    (547,204) -- (569.15,173.87) ;
\draw [shift={(570.33,172.26)}, rotate = 126.32] [color={rgb, 255:red, 0; green, 0; blue, 0 }  ][line width=0.75]    (10.93,-3.29) .. controls (6.95,-1.4) and (3.31,-0.3) .. (0,0) .. controls (3.31,0.3) and (6.95,1.4) .. (10.93,3.29)   ;
\draw    (567,204) -- (589.15,173.87) ;
\draw [shift={(590.33,172.26)}, rotate = 126.32] [color={rgb, 255:red, 0; green, 0; blue, 0 }  ][line width=0.75]    (10.93,-3.29) .. controls (6.95,-1.4) and (3.31,-0.3) .. (0,0) .. controls (3.31,0.3) and (6.95,1.4) .. (10.93,3.29)   ;
\draw    (587,204) -- (609.15,173.87) ;
\draw [shift={(610.33,172.26)}, rotate = 126.32] [color={rgb, 255:red, 0; green, 0; blue, 0 }  ][line width=0.75]    (10.93,-3.29) .. controls (6.95,-1.4) and (3.31,-0.3) .. (0,0) .. controls (3.31,0.3) and (6.95,1.4) .. (10.93,3.29)   ;
\draw    (607,204) -- (629.15,173.87) ;
\draw [shift={(630.33,172.26)}, rotate = 126.32] [color={rgb, 255:red, 0; green, 0; blue, 0 }  ][line width=0.75]    (10.93,-3.29) .. controls (6.95,-1.4) and (3.31,-0.3) .. (0,0) .. controls (3.31,0.3) and (6.95,1.4) .. (10.93,3.29)   ;
\draw    (627,204) -- (649.15,173.87) ;
\draw [shift={(650.33,172.26)}, rotate = 126.32] [color={rgb, 255:red, 0; green, 0; blue, 0 }  ][line width=0.75]    (10.93,-3.29) .. controls (6.95,-1.4) and (3.31,-0.3) .. (0,0) .. controls (3.31,0.3) and (6.95,1.4) .. (10.93,3.29)   ;

\draw    (253.33,147.59) .. controls (226.6,170.36) and (174.39,120.61) .. (131.63,161.33) ;
\draw [shift={(130.33,162.59)}, rotate = 315] [color={rgb, 255:red, 0; green, 0; blue, 0 }  ][line width=0.75]    (10.93,-3.29) .. controls (6.95,-1.4) and (3.31,-0.3) .. (0,0) .. controls (3.31,0.3) and (6.95,1.4) .. (10.93,3.29)   ;
\draw    (295,148.59) .. controls (309.09,171.36) and (364.22,126.51) .. (405.1,157.63) ;
\draw [shift={(406.33,158.59)}, rotate = 218.83] [color={rgb, 255:red, 0; green, 0; blue, 0 }  ][line width=0.75]    (10.93,-3.29) .. controls (6.95,-1.4) and (3.31,-0.3) .. (0,0) .. controls (3.31,0.3) and (6.95,1.4) .. (10.93,3.29)   ;

\draw (76,23) node [anchor=north west][inner sep=0.75pt]  [font=\Large]  {$\hat{H}$};
\draw (263, 36) node [anchor=north west][inner sep=0.75pt]  [font=\Large]  {$=$};
\draw (255,78.4) node [anchor=north west][inner sep=0.75pt]  [font=\Large]  {$\Downarrow $};
\draw (296,29.4) node [anchor=north west][inner sep=0.75pt]  [font=\Large]  {$E$};
\draw (260,113.4) node [anchor=north west][inner sep=0.75pt]  [font=\Large]  {$\hat{H}$};
\draw (0,168) node [anchor=north west][inner sep=0.75pt]  [font=\large]  {$\hat{P}_{i,j}$};
\draw (189.33, 181) node [anchor=north west][inner sep=0.75pt]  [font=\Large]  {$=$};
\draw (245,165) node [anchor=north west][inner sep=0.75pt][font=\large]  {$\sum_i\hat{h}_{i}$};
\draw (451.33, 181) node [anchor=north west][inner sep=0.75pt]  [font=\Large]  {$=$};
\draw (476,174.4) node [anchor=north west][inner sep=0.75pt]  [font=\large]  {$E$};
\draw (220,176.4) node [anchor=north west][inner sep=0.75pt]  [font=\large]  {$0$};
\end{tikzpicture}
    \label{fig:ponchi}
    \caption{A graphical representation of Thm\,\ref{thm:main informal}. If a local Hamiltonian supports ferromagnetic scar states as exact eigenstates (top), then it must admit decomposition into an annihilator and a Zeeman term. The annihilator consists of local projectors that annihilate scars locally (bottom left), and the Zeeman term implies the equally spaced scar tower (bottom right).}
\end{figure}
The goal of this paper is to clarify the relation between these two empirical features: the existence of the ferromagnetic scar states and the decomposition of the Hamiltonian into an annihilator and a Zeeman term, which, \textit{a priori}, need not be connected. A first step in this direction was taken recently for the $W$ state (the spin-1/2 single-magnon state) via an explicit parent-Hamiltonian analysis\,\cite{gioia2025}. Here we extend that perspective to a broad class of QMBS Hamiltonians. Specifically, we prove the following theorem (informally),
\begin{theorem}[Main theorem (informal)]\label{thm:main informal}
If a local Hamiltonian hosts ferromagnetic scar states as exact eigenstates, then it necessarily admits a decomposition into a Zeeman term and an operator that annihilates the scar manifold; moreover, the annihilating operator is built from local projectors that annihilate the scar states locally.
\end{theorem}
This structure (annihilator\,$+$\,Zeeman) is a natural generalization of the Shiraishi--Mori construction\,\cite{mori_shiraishi}, and our result shows that this construction is essentially \textit{exhaustive} for a broad class of QMBS. The proof of Thm\,\ref{thm:main informal} shares its mathematical origin with the commutant-algebraic viewpoint on QMBS\,\cite{moudgalya2023} and relies on elementary representation theory of the symmetric group. Note that the question of when the annihilator admits a fully strictly local decomposition in complete generality remains open, although we also establish such a refinement under additional assumptions.

The remainder of this paper is organized as follows. In Sec.\,\ref{sec:XY}, we review the spin-1 XY model, a paradigmatic model of exactly solvable QMBS, and use it to motivate the central questions. In Sec.\,\ref{sec:setup}, we introduce a general setup that abstracts the key structure of the spin-1 XY model. Sec.\,\ref{sec:preliminaries} reviews the representation-theoretic ingredients that will be used throughout. Secs\,\ref{sec:HA} and \ref{sec:locality} contain the main structural results: Sec.\,\ref{sec:HA} shows that any operator annihilating the ferromagnetic scar manifold must be expressible in terms of local projectors that annihilate it locally, and Sec.\,\ref{sec:locality} proves that imposing locality restricts the remaining, non-annihilating part of the Hamiltonian to a Zeeman term. Sec.\,\ref{sec:quasi-local decomp} discusses sufficient conditions for obtaining strictly local annihilators and connects them to Dzyaloshinskii--Moriya (DM)-type interactions. We conclude in Sec.\,\ref{sec:discussion} with a discussion and outlook.

\section{Motivation: the spin-1 XY model}\label{sec:XY}
The spin-1 XY model provides a paradigmatic example of a non-integrable model hosting analytically tractable ferromagnetic scar states. On a one-dimensional chain of even length $L$ with periodic boundary conditions it is defined by,
\begin{equation}
    \hat{H}_{\mr{XY}}\coloneqq J\sum_{x=1}^L\big(\hat{S}_x^1\hat{S}_{x+1}^1+\hat{S}_x^2\hat{S}_{x+1}^2\big)+\sum_{x=1}^L\big(h\hat{S}_x^3+D(\hat{S}_x^3)^2\big),
\end{equation}
where $\hat{S}_x^\alpha\,(\alpha=1,2,3)$ are spin-1 operators satisfying the usual $\mf{su}(2)$ commutation relation. The local Hilbert space at each site is three-dimensional, $\mf{h}_x\cong\mb{C}^3$, and the many-body Hilbert space is the tensor product, $\mc{H}=\bigotimes_{x=1}^L\mf{h}_x$. At least for $D=0$ (isotropic limit), one can rigorously show that there are no non-trivial local conserved quantities beyond the Hamiltonian and the total magnetization $\hat{S}^3_{\mr{tot}}\coloneqq\sum_{x=1}^L\hat{S}_x^3$; hence the model is non-integrable\,\cite{KeitaOmiyaPhD}\footnote{There exists another non-trivial $SU(2)$ symmetry constructed from non-local generators\,\cite{odea2024levelstatistics}. The support of the generators spans the entire chain and thus this symmetry is not considered to be ``local'' in our discussion.}. We expect this non-integrability to persist for generic $D$.

Despite this non-integrability, the model admits an exact tower of low-entanglement eigenstates $\{\ket*{S^\pi_n}\}_{n=0}^L$. Let $\ket*{\pm}$ and $\ket*{0}$ denote eigenstates of $\hat{S}^3$ with $\hat{S}^3\ket*{\pm}=\pm\ket*{\pm}$ and $\hat{S}^3\ket*{0}=0$. Then the scar states are given by,
\begin{equation}
    \ket*{S^\pi_n}\coloneqq\frac{1}{\mc{N}_n}\left(\sum_{x=1}^L(-1)^x(\hat{S}_x^+)^2\right)^n\bigotimes_{x=1}^L\ket*{-}_x,
\end{equation}
where $\hat{S}_x^+\coloneqq\hat{S}_x^1+i\hat{S}_x^2$ and $\mc{N}_n$ is a normalization factor. To align these states with our notion of ferromagnetic scar states, it is convenient to remove the ``$\pi$-momentum'' by a site-dependent unitary transformation $\hat{U}$,
\begin{equation}
    \begin{split}
        &\hat{U}\coloneqq\exp\left[-i\pi\sum_{x=1}^Lx\dyad*{+}_x\right]\\
        &\hat{U}\ket*{S^\pi_n}\eqqcolon\ket*{S_n}=\frac{1}{\mc{N}_n}\left(\sum_{x=1}^L(\hat{S}_x^+)^2\right)^n\bigotimes_{x=1}^L\ket*{-}_x.
    \end{split}
\end{equation}

The state $\ket*{S_n}$ has two key features. First, their local configurations are restricted to the two-dimensional subspace $\mf{h}_x^s\subset\mf{h}_x$, spanned by $\ket*{+}$ and $\ket*{-}$, so that $\ket*{S_n}\in\bigotimes_{x\in\Lambda}\mf{h}_x^s$. Second, within the reduced Hilbert space, the family $\{\ket*{S_n}\}_{n=0}^L$ spans the totally symmetric sector, $\mr{Sym}^L(\mf{h}^s)$, i.e., the spin-$L/2$ multiplet.   

The mechanism behind these exact eigenstates becomes transparent once we examine the conjugated Hamiltonian. One can decompose it as,
\begin{equation}
    \hat{U}\hat{H}_{\mr{XY}}\hat{U}^\dag=\sum_{x=1}^L\hat{h}_{x,x+1}^{(1)}\hat{P}_{x,x+1}^{(1)}+\sum_{x=1}^L\hat{h}_{x,x+1}^{(2)}\hat{P}_{x,x+1}^{(2)}+\sum_{x=1}^L\hat{h}_x,
\end{equation}
where $\hat{h}_{x,x+1}^{(i)}\,(i=1,2)$ and $\hat{h}_x$ are local terms, 
\begin{equation}
    \begin{split}
        \hat{h}_{x,x+1}^{(1)}&\coloneqq J\Big(-\dyad*{0,+}{+,0}-\dyad*{+,0}{0,+}+\dyad*{0,-}{-,0}+\dyad*{-,0}{0,-}\Big)_{x,x+1}\\
        &+(-1)^xJ\big(\ket*{+,-}-\ket*{-,+}\big)\bra*{0,0}_{x,x+1}\\
        \hat{h}_{x,x+1}^{(2)}&\coloneqq(-1)^xJ\ket*{0,0}\big(\bra*{+,-}-\bra*{-,+}\big)_{x,x+1}\\
        \hat{h}_x&\coloneqq h\hat{S}_x^3+D(\hat{S}_x^3)^2,
    \end{split}
\end{equation}
and $\hat{P}_{x,x+1}^{(i)}\,(i=1,2)$ is a projector,
\begin{equation}\begin{split}
    &\hat{P}_{x,x+1}^{(1)}\coloneqq\dyad{0}_x+\dyad*{0}_{x+1}-\dyad*{0,0}_{x,x+1}\\
    &\hat{P}_{x,x+1}^{(2)}\coloneqq\frac{1}{2}\big(\ket*{+,-}-\ket*{-,+}\big)\big(\bra*{+,-}-\bra*{-,+}\big)_{x,x+1}.
    \end{split}
\end{equation}
The projector $\hat{P}_{x,x+1}^{(1)}$ annihilates any state with no $\ket*{0}$ component, and thus it trivially annihilates the entire family $\{\ket*{S_n}\}_{n=0}^L$. The second projector $\hat{P}_{x,x+1}^{(2)}$ annihilates the triplets sector built from $\ket*{\pm}$. Since $\{\ket*{S_n}\}_{n=0}^L$ forms a totally symmetric multiplet, any two-site reduced density matrix lies in the triplet subspace and is therefore annihilated by $\hat{P}_{x,x+1}^{(2)}$. The remaining on-site term $\hat{h}_x$ is a Zeeman term and yields an equidistant energy spectrum within the tower $\{\ket*{S_n}\}_{n=0}^L$.

In summary, the spin-1 XY model hosts a family of ferromagnetic scar states $\{\ket*{S_n^\pi}\}_{n=0}^L$, which can be viewed as the totally symmetric wavefunctions (Dicke-type states) living in the tensor product of reduced on-site subspaces $\bigotimes_{x=1}^L\mf{h}_x^s\cong\mb{C}^{2\otimes L}$. The key structural point is that the Hamiltonian admits a decomposition into local projector terms  $\hat{P}_{x,x+1}^{(1)}$ and $\hat{P}_{x,x+1}^{(2)}$ that locally annihilate the scar manifold, and a Zeeman-type term that produces an equidistant ladder of energies. This is precisely the hallmark of the Shiraishi--Mori (SM) construction\,\cite{mori_shiraishi}, in which the local interactions are built from projectors that vanish on a designed subspace. 

In the remainder of the paper, we will generalize the spin-1 XY model and discuss the converse of the SM construction: we will show that the existence of ferromagnetic, totally symmetric exact eigenstates \textit{forces} the presence of the corresponding local projector structure in the Hamiltonian. Consequently, a (slightly generalized) SM construction is, in this sense, exhaustive for ferromagnetic scar states.

\section{Setup}\label{sec:setup}
\subsection*{Hilbert space and Hamiltonian}
We consider a lattice $\Lambda$ with $|\Lambda|=N$ sites labeled by $x$. Each site carries a $d$-dimensional local Hilbert space $\mf{h}_x\cong\mb{C}^d$, so the many-body Hilbert space is the tensor product $\mc{H}=\bigotimes_{x\in\Lambda}\mf{h}_x$. We assume that a local Hamiltonian $\hat{H}$ can be written as a sum of operators supported on connected subregions (subgraphs) $X\subset\Lambda$,
\begin{equation}\label{eq:Hamiltonian}
    \hat{H}=\sum_{X\subset\Lambda}\hat{h}_X,
\end{equation}
where $\hat{h}_X$ acts non-trivially only on the sites in $X$.




\subsection*{Ferromagnetic scar states}
We now define a general class of ``ferromagnetic scar states'' that abstracts the structure seen in the spin-1 XY example. We fix an arbitrary on-site target space $\mf{h}^s=\mb{C}^{d_s}(\subset\mf{h})$. Any operator acting non-trivially only within $\mf{h}^s$ can be seen as an element of $\mf{su}(\mf{h}^s)$ (with $\mf{h}^s$ carrying its fundamental representation). Accordingly, we choose a Cartan subalgebra $\{\hat{H}^i\}_{i=1}^{d_s-1}$ and ladder operators $\{\hat{E}^\alpha\}_\alpha$. The target space $\mf{h}^s$ is spanned by simultaneous eigenstates $\ket*{\vec{m}}$ (with $\vec{m}$ a weight vector) satisfying 
\begin{equation}
    \hat{H}^i\ket*{\vec{m}}=m^i\ket*{\vec{m}},
\end{equation}
where $m^i$ is the $i$th component of the weight vector $\vec{m}=(m^1,\cdots,m^{d_s-1})$. In this language, the ferromagnetic scar manifold will be identified with the totally symmetric representation $\mr{Sym}^N(\mf{h}^s)$ generated from $\mf{su}(\mf{h}^s)$. 

Concretely, we introduce collective generators obtained by summing the on-site $\mf{su}(\mf{h}^s)$ operators over $\Lambda$:
\begin{equation}
    \hat{H}_{\mr{tot}}^i\coloneqq\sum_{x\in\Lambda}\hat{H}^i_x,\,\,\hat{J}^\alpha\coloneqq\sum_{x\in\Lambda}\hat{E}_x^\alpha,
\end{equation}
where the subscript $x$ indicates an operator acting non-trivially only on site $x$. The ferromagnetic scar states are then identified with the $N$th symmetric power of the fundamental representation. We define a convenient weight basis $\{\ket*{\psi^\Lambda_{\vec{m}}}\}_{\vec{m}}$ for $\mr{Sym}^N(\mf{h}^s)$ in the following manner. Let $\ket*{\Omega}_x$ be the highest-weight state on site $x$, and let $\Delta$ be a set of simple roots. For each weight $\vec{m}$, choose nonnegative integers $\{p_\alpha\}_{\alpha\in\Delta}$ that generate $\vec{m}$ from the highest weight via collective lowering operations\footnote{Note that the integers $p_\alpha\geq0$ are chosen so that the resulting state has total weight $\vec{m}$. This choice and the ordering of the product over $\alpha$ are not unique in general. However, whenever we need an unambiguous expression, we will use the equivalent symmetrized-configuration form in Eq.\,\eqref{eq:permutation expression}.}. We then set,
\begin{equation}\label{eq:Cartan basis state}\begin{split}
    &\ket*{\psi^\Lambda_{\vec{m}}}\coloneqq\frac{1}{\mc{N}_{\vec{m}}}\prod_{\alpha\in\Delta}\big(\hat{J}^\alpha\big)^{p_\alpha}\bigotimes_{x\in\Lambda}\ket*{\Omega}_x\\
    &\hat{H}^i_{\mr{tot}}\ket*{\psi^\Lambda_{\vec{m}}}=m^i\ket*{\psi^\Lambda_{\vec{m}}},
    \end{split}
\end{equation}
where $\mc{N}_{\vec{m}}$ is a normalization constant.

An equivalent and often more intuitive representation is obtained by symmetrizing a reference product configuration. Let $(\vec{m}_x)_{x\in\Lambda}$ be an assignment of on-site weight satisfying the global constraint,
\begin{equation}
    \sum_{x\in\Lambda}m_x^i=m^i\,\,\text{for }\forall i.
\end{equation}
Then the corresponding symmetric state can be written as,
\begin{equation}\label{eq:permutation expression}
    \ket*{\psi^\Lambda_{\vec{m}}}=\frac{1}{\mc{N}'_{\vec{m}}}\sum_{\sigma\in\mf{S}_N}\hat{\sigma}\bigotimes_{x\in\Lambda}\ket*{\vec{m}_x}_x,
\end{equation}
where $\mc{N}'_{\vec{m}}$ is a normalization constant and $\hat{\sigma}$ denotes the permutation operator acting by,
\begin{equation}
    \hat{\sigma}\bigotimes_{x\in\Lambda}\ket{\vec{m}_x}_x=\bigotimes_{x\in\Lambda}\ket*{\vec{m}_{\sigma(x)}}_x.
\end{equation}
This makes manifest that $\ket*{\psi^\Lambda_{\vec{m}}}$ lies in the totally symmetric sector.

Before proceeding to the representation-theoretic preliminaries, let us summarize the logic of the proof. We first show that any operator annihilating the ferromagnetic scar manifold must contain local projector factors that annihilate the symmetric sector locally (Sec.\,\ref{sec:HA}). We then show that, once locality is imposed and the full totally symmetric weight basis consists of exact eigenstates, the only non-annihilating contribution within the scar manifold is a Zeeman term (Sec.\,\ref{sec:locality}). This yields the generalized Shiraishi--Mori decomposition stated informally in Sec\,\ref{sec:intro}.

\section{Preliminaries}\label{sec:preliminaries}
In this section we recapitulate several standard representation-theoretic tools that will be used repeatedly in the subsequent analysis. A main objective is to isolate the ferromagnetic scar manifold as the totally symmetric sector of the permutation action and to characterize operators that can act nontrivially within that sector. For the arguments in Secs.~\ref{sec:HA} and \ref{sec:locality}, we only need three ingredients: (i) the Young-symmetrizer resolution of the permutation action on $V^{\otimes N}$, (ii) the projector onto the totally symmetric sector, and (iii) the characterization of the permutation commutant in terms of collective generators. We therefore restrict ourselves to the minimal representation-theoretic input needed later.

\subsection{Permutation resolution of the Hilbert space}\label{subsec:rep of SN}

\begin{figure}
\centering

\tikzset{every picture/.style={line width=0.75pt}} 

\begin{tikzpicture}[x=0.75pt,y=0.75pt,yscale=-1,xscale=1]

\draw   (30,11) -- (80,11) -- (80,61) -- (30,61) -- cycle ;
\draw   (80,11) -- (130,11) -- (130,61) -- (80,61) -- cycle ;
\draw   (130,11) -- (180,11) -- (180,61) -- (130,61) -- cycle ;
\draw   (30,61) -- (80,61) -- (80,111) -- (30,111) -- cycle ;
\draw   (80,61) -- (130,61) -- (130,111) -- (80,111) -- cycle ;
\draw   (261,11) -- (311,11) -- (311,61) -- (261,61) -- cycle ;
\draw   (311,11) -- (361,11) -- (361,61) -- (311,61) -- cycle ;
\draw   (361,11) -- (411,11) -- (411,61) -- (361,61) -- cycle ;
\draw   (261,61) -- (311,61) -- (311,111) -- (261,111) -- cycle ;
\draw   (311,61) -- (361,61) -- (361,111) -- (311,111) -- cycle ;
\draw    (436.33,39) -- (417,39) ;
\draw [shift={(415,39)}, rotate = 360] [color={rgb, 255:red, 0; green, 0; blue, 0 }  ][line width=0.75]    (10.93,-3.29) .. controls (6.95,-1.4) and (3.31,-0.3) .. (0,0) .. controls (3.31,0.3) and (6.95,1.4) .. (10.93,3.29)   ;
\draw    (437.33,91) -- (418,91) ;
\draw [shift={(416,91)}, rotate = 360] [color={rgb, 255:red, 0; green, 0; blue, 0 }  ][line width=0.75]    (10.93,-3.29) .. controls (6.95,-1.4) and (3.31,-0.3) .. (0,0) .. controls (3.31,0.3) and (6.95,1.4) .. (10.93,3.29)   ;
\draw    (284.33,133.67) -- (284.33,117.67) ;
\draw [shift={(284.33,115.67)}, rotate = 90] [color={rgb, 255:red, 0; green, 0; blue, 0 }  ][line width=0.75]    (10.93,-3.29) .. controls (6.95,-1.4) and (3.31,-0.3) .. (0,0) .. controls (3.31,0.3) and (6.95,1.4) .. (10.93,3.29)   ;
\draw    (334.33,134.67) -- (334.33,118.67) ;
\draw [shift={(334.33,116.67)}, rotate = 90] [color={rgb, 255:red, 0; green, 0; blue, 0 }  ][line width=0.75]    (10.93,-3.29) .. controls (6.95,-1.4) and (3.31,-0.3) .. (0,0) .. controls (3.31,0.3) and (6.95,1.4) .. (10.93,3.29)   ;

\draw (279,28.4) node [anchor=north west][inner sep=0.75pt]    {$1$};
\draw (281,76.4) node [anchor=north west][inner sep=0.75pt]    {$2$};
\draw (330,28.4) node [anchor=north west][inner sep=0.75pt]    {$3$};
\draw (381,28.4) node [anchor=north west][inner sep=0.75pt]    {$4$};
\draw (330,77.4) node [anchor=north west][inner sep=0.75pt]    {$5$};
\draw (436,30.4) node [anchor=north west][inner sep=0.75pt]    {$\mathfrak{S}_{3}$};
\draw (437,81.4) node [anchor=north west][inner sep=0.75pt]    {$\mathfrak{S}_{2}$};
\draw (275,137.4) node [anchor=north west][inner sep=0.75pt]    {$\mathfrak{S}_{2}$};
\draw (324,137.4) node [anchor=north west][inner sep=0.75pt]    {$\mathfrak{S}_{2}$};

\end{tikzpicture}
\caption{(Left) a partition of $5$ into $(3,2)$, which corresponds to an irrep of $\mf{S}_5$. (Right) the construction of the Young symmetrizer corresponding to $(3,2)$. The subgroup $\mf{P}_{(3,2)}$, which preserves each row, is isomorphic to $\mf{S}_3\times\mf{S}_2$, permuting within $\{1,3,4\}$ and $\{2,5\}$. Similarly, the subgroup $\mf{Q}_{(3,2)}$, which preserves each column, is isomorphic to $\mf{S}_2\times\mf{S}_2(\times\mf{S}_1)$, permuting within $\{1,2\}$ and $\{3,5\}$.}
\label{fig:tableau}
\end{figure}

Let $V$ be a finite-dimensional vector space spanned by $\{\ket{e_i}\}_i$ and consider $V^{\otimes N}$ with the natural right action of the symmetric group $\mf{S}_N$:
\begin{equation}
    \hat{\sigma}^R\ket{e_1}\otimes\cdots\otimes\ket{e_N}\coloneqq\ket*{e_{\sigma^{-1}(1)}}\otimes\cdots\otimes\ket*{e_{\sigma^{-1}(N)}}.
\end{equation}
The superscript $R$ emphasizes that this is a right action.  We also define the corresponding left action by
\begin{equation}
    \hat{\sigma}\ket{e_1}\otimes\cdots\otimes\ket{e_N}\coloneqq\ket*{e_{\sigma(1)}}\otimes\cdots\otimes\ket*{e_{\sigma(N)}}.
\end{equation}
These are related by $\hat{\sigma}^R=\hat{\sigma}^{-1}$. Importantly, the right action satisfies $\hat{\sigma}^R\hat{\tau}^R=\widehat{\sigma\tau}^R$, whereas the left action composes in the opposite order ($\what{\sigma\tau}=\hat{\tau}\hat{\sigma}$).

Irreducible representations (irreps) of $\mf{S}_N$ are labeled by Young diagrams $\lambda\vdash N$. In particular, the trivial representation corresponds to the one-row diagram $(N)$, i.e., to the totally symmetric sector. For each standard tableau $T$ with shape $\lambda$, let $c_T\in\mb{C}[\mf{S}_N]$ denote the corresponding Young symmetrizer (see below and Appendix\,\ref{app:SN}). Then the right action of $\mathbb{C}[\mf{S}_N]$ yields a decomposition
\begin{equation}
V^{\otimes N}=\bigoplus_T\hat{c}^{R}_T V^{\otimes N}\eqqcolon\bigoplus_T\mb{S}_TV,
\label{eq:young_resolution_condensed}
\end{equation}
where $T$ runs over standard tableaux. The projector onto the totally symmetric sector is the normalized symmetrizer
\begin{equation}
\hat{c}^{R}_{(N)}=\frac{1}{|\mf{S}_N|}\sum_{\sigma\in \mf{S}_N}\hat{\sigma}^{R}=\frac{1}{|\mf{S}_N|}\sum_{\sigma\in \mf{S}_N}\hat{\sigma}.
\label{eq:sym_projector_condensed}
\end{equation}
In our setting, the ferromagnetic scar manifold is precisely the sector $\mr{Sym}^N(\mf{h}^s)\subset(\mf{h}^s)^{\otimes N}$ selected by $\hat{c}^{R}_{(N)}$.

For a generic standard tableau $T$, the Young symmetrizer $c_T$ is a product of the symmetrizer $a_T$ and the anti-symmetrizer $b_T$, i.e., $c_T=a_Tb_T$: let $P_T$ and $Q_T$ denote the subgroups of $\mf{S}_N$ that preserve each row and each column, respectively:
\begin{equation}\begin{split}
P_T&\coloneqq\{\sigma\in\mf{S}_N; \sigma \text{ preserves each row}\}\\
Q_T&\coloneqq\{\sigma\in\mf{S}_N; \sigma \text{ preserves each column}\},
\end{split}
\end{equation}
which yield the symmetrizer and anti-symmetrizer by
\begin{equation}\begin{split}
a_T&\coloneqq\sum_{\sigma\in P_T}\sigma\\
b_T&\coloneqq\sum_{\sigma\in Q_T}\mathrm{sgn}(\sigma)\sigma,
\end{split}
\end{equation}

\subsection{Commutant and collective generators}

Let $\mc{H}$ be a right $G$-module, and define its commutant $\mr{Comm}(G)$ by, 
\begin{equation}
    \mr{Comm}(G)\coloneqq\left\{\varphi\in\mc{L}(\mc{H});\,\varphi(ug)=\varphi(u)g\,\,\text{for }\forall u\in\mc{H},\,\forall g\in G\right\}.
\end{equation}
Namely, $\varphi\in\mr{Comm}(G)$ commutes with all symmetry operators in $G$.

We recall a characterization of the commutant that will be central in later sections:
\begin{theorem}\label{thm:GL(V)}
    The commutant of the $\mf{S}_N$ action on $V^{\otimes N}$ is generated by tensor powers of single-site unitaries $U(V)$:
    \begin{equation}\label{eq:Comm(SN)}
    \mr{Comm}(\mf{S}_N)\cong\left\langle\hat{U}^{\otimes N};\,\,\hat{U}\in U(V)\right\rangle,
    \end{equation}
    where $\expval{\bullet}$ is the algebra generated by its argument.
\end{theorem}
We defer the proof to Appendix\,\ref{proof:GL(V)}. This theorem identifies the multiplicity spaces $\mb{S}_TV$ 
with the commutant sector, and in the QMBS setting it underlies the familiar interpretation of ferromagnetic scar families as $SU(2)$-multiplets.

Eq.\,\eqref{eq:Comm(SN)} implies that any operator $\hat{O}$ commuting with all permutations admits an expansion in collective $\mf{su}(V)$ generators:
\begin{equation}\label{eq:collective su(V)}
    \hat{O}=\sum_{l=0}^N\sum_{\mu_1\cdots\mu_l}g_{\mu_1\cdots\mu_l}\hat{T}^{\mu_1}\cdots\hat{T}^{\mu_l},\,\,\hat{T}^\mu\coloneqq\sum_{x=1}^N\hat{T}_x^{\mu}
\end{equation}
where $\hat{T}^\mu_x$ is a single-site $\mf{su}(V)$ generator acting only on site $x$.

Moreover, for a general operator $\hat{O}\in\mc{L}(V^{\otimes N})$, we can project onto $\mr{Comm}(\mf{S}_N)$ via symmetrization,
\begin{equation}\label{eq:symmetrization}
    \mc{S}(\hat{O})\coloneqq\frac{1}{|\mf{S}_N|}\sum_{\sigma\in\mf{S}_N}\hat{\sigma}\hat{O}\hat{\sigma}^{-1}.
\end{equation}
By Thm.\,\ref{thm:GL(V)}, $\mc{S}(\hat{O})$ is again expressible as a polynomial in collective $\mf{su}(V)$ generators.

\section{Annihilation implies local projectors}\label{sec:HA}
We begin by analyzing annihilators, i.e., operators on the full many-body Hilbert space that annihilate the totally symmetric subspace $\mr{Sym}^N(\mf{h}^s)$. The key point of this section is that any such annihilator necessarily contains strictly local projectors, possibly multiplied by non-Hermitian operators with long range, so that the annihilation can always be attributed to a local ``forbidden component''.
\begin{lemma}\label{lemma:main}
    If $\hat{O}\in\mc{L}(\mc{H})$ ($\mc{H}=\bigotimes_{x\in\Lambda}\mf{h}_x$) annihilate $\mr{Sym}^N(\mf{h}^s)$, then $\hat{O}$ can be written as a sum of terms each containing strictly local projectors:
    \begin{equation}\label{eq:lemma main}        \hat{O}=\sum_x\hat{o}^{(1)}_{[x]}\hat{P}_x+\sum_{\expval{x,y}}\hat{o}^{(2)}_{[xy]}\hat{P}_{xy},
    \end{equation}
    where $\expval{x,y}$ denotes nearest-neighbor sites, the projectors $\hat{P}_x$ and $\hat{P}_{xy}$ act non-trivially only on site $x$ and on the pair $(x,y)$, respectively. Each projector annihilates $\mr{Sym}^N(\mf{h}^s)$.
\end{lemma}
\begin{remark}
    Eq.\,\eqref{eq:lemma main} is an existence statement. The decomposition constructed below may be far from optimal: it should be viewed as a ``worst-case'' bound that certifies locality of the projector factors, not efficiency of the expansion.
\end{remark}
\begin{remark}
    While $\hat{P}_x$ and $\hat{P}_{xy}$ are strictly local (support on one site or two neighboring sites), the coefficient operators  $\hat{o}^{(1)}_{[x]}$ and $\hat{o}^{(2)}_{[xy]}$ are not assumed local and may be long-ranged. In the following, a bracket $[x]$ (or $[xy]$) in subscripts indicates a label of operators and does not necessarily imply strict locality (i.e., the support of operators is in general larger).
\end{remark}

This lemma already captures the structural core of the SM-type mechanism: any operator that vanishes on the ferromagnetic scar manifold does so because it contains local ``filters'' that remove the symmetric sector. 

Our proof proceeds as follows: we first treat the case $\mf{h}^s=\mf{h}$ (the scar manifold is the totally symmetric subspace of the full on-site Hilbert space). Once this is established, the general case $\mf{h}^s\subset\mf{h}$ follows by a short reduction.

\begin{proof}
\textbf{Step 1: the case $\mf{h}^s=\mf{h}$.} 

We assume that $\hat{O}$ annihilates $\mr{Sym}^N(\mf{h})$. Using the Young symmetrizer decomposition of $\mc{H}$ (with $(N)$ denoting the one-row Young diagram), we have 
\begin{equation}
    \hat{O}\sum_T\hat{c}^R_T=\hat{O}\sum_{T\not=(N)}\hat{c}_T^R=\sum_{T\not=(N)}\hat{O}\hat{c}_T^R.
\end{equation}
Hence it suffices to show that, for each $T\not=(N)$, the operator $\hat{O}\hat{c}_T^R$ can be expanded into terms containing strictly local projectors. 

To this end, we inspect $\hat{c}_T^R$ and the corresponding tableau $T$: since $T$ has at least two rows, its column-preserving subgroup $\mf{Q}_T$ is non-trivial and factorizes into symmetric groups associated with columns of height$\geq2$ (also see the right panel of Fig.\,\ref{fig:tableau}),
\begin{equation}
    \mf{Q}_T\cong\mf{S}_{\lambda'_1}\times\cdots\times\mf{S}_{\lambda'_M},
\end{equation}
where $M$ is the number of columns of length larger than one.

For each non-trivial factor $\mf{S}_{\lambda'_j}$, one can choose a transposition $t_{a_j,b_j}$ and decompose the group into cosets of its alternating subgroup $\mf{A}_{\lambda'_j}$:
\begin{equation}
    \mf{S}_{\lambda'_j}=\mf{A}_{\lambda'_j}\cup\mf{A}_{\lambda'_j}t_{a_j,b_j},
\end{equation}
This yields a convenient factorization of the anti-symmetrizing part $b_T$ of the Young symmetrizer,
\begin{equation}
    b_T=\prod_{j=1}^M\left(\sum_{\sigma\in\mf{A}_{\lambda'_j}}\sigma\big(1-t_{a_j,b_j}\big)\right).
\end{equation}
Because the factors in the product commute, we can regroup terms and re-write $b_T$ as
\begin{equation}\label{eq:antisym}
    b_T=\sideset{}{'}\sum_\sigma\mr{sgn}(\sigma)\sigma\left(1-t_{a_1,b_1}\right),
\end{equation}
where $\sum\nolimits'_\sigma$ denotes the sum over the subgroup $\mf{A}_{\lambda'_1}\times\mf{S}_{\lambda'_2}\times\cdots\times\mf{S}_{\lambda'_M}$. Combining with the symmetrizing part $a_T=\sum_{\tau\in\mf{P}_T}\tau$, the Young symmetrizer becomes
\begin{equation}
    c_T=\left(\sum_{\tau\in\mf{P}_T}\tau\right)\left(\sideset{}{'}\sum_\sigma\mr{sgn}(\sigma)\sigma(1-t_{a_1,b_1})\right).
\end{equation}
Translating this group algebra identity into the right action on $\mc{H}$ (using $\hat{\sigma}^R=\hat{\sigma}^{-1}$), we obtain the corresponding operator identity,
\begin{equation}
    \hat{c}_T^R=\sum_{\tau\in\mf{P}_T}\sideset{}{'}\sum_\sigma\mr{sgn}(\sigma)\hat{\tau}^{-1}\hat{\sigma}^{-1}(1-\widehat{\mr{SWAP}}_{a_1,b_1}),
\end{equation}
where $\what{\mr{SWAP}}_{a,b}$ implements the site swap between sites $a$ and $b$,
\begin{equation}
    \what{\mr{SWAP}}_{a,b}\ket*{e_a}_a\otimes\ket*{e_b}_b=\ket*{e_b}_a\otimes\ket*{e_a}_b.
\end{equation}
The factor $1-\widehat{\mr{SWAP}}_{a_1,b_1}$ is proportional to the projector $\hat{P}_{a_1,b_1}$ onto the antisymmetric subspace under the site swap $a_1\leftrightarrow b_1$.

\begin{figure}
    \centering
    \tikzset{every picture/.style={line width=0.75pt}} 
\begin{tikzpicture}[x=0.75pt,y=0.75pt,yscale=-1.2,xscale=1.2]
\draw  [fill={rgb, 255:red, 0; green, 0; blue, 0 }  ,fill opacity=1 ] (101,145.5) .. controls (101,143.01) and (103.01,141) .. (105.5,141) .. controls (107.99,141) and (110,143.01) .. (110,145.5) .. controls (110,147.99) and (107.99,150) .. (105.5,150) .. controls (103.01,150) and (101,147.99) .. (101,145.5) -- cycle ;
\draw  [fill={rgb, 255:red, 0; green, 0; blue, 0 }  ,fill opacity=1 ] (141,145.5) .. controls (141,143.01) and (143.01,141) .. (145.5,141) .. controls (147.99,141) and (150,143.01) .. (150,145.5) .. controls (150,147.99) and (147.99,150) .. (145.5,150) .. controls (143.01,150) and (141,147.99) .. (141,145.5) -- cycle ;
\draw  [fill={rgb, 255:red, 0; green, 0; blue, 0 }  ,fill opacity=1 ] (181,145.5) .. controls (181,143.01) and (183.01,141) .. (185.5,141) .. controls (187.99,141) and (190,143.01) .. (190,145.5) .. controls (190,147.99) and (187.99,150) .. (185.5,150) .. controls (183.01,150) and (181,147.99) .. (181,145.5) -- cycle ;

\draw  [fill={rgb, 255:red, 0; green, 0; blue, 0 }  ,fill opacity=1 ] (101,185.5) .. controls (101,183.01) and (103.01,181) .. (105.5,181) .. controls (107.99,181) and (110,183.01) .. (110,185.5) .. controls (110,187.99) and (107.99,190) .. (105.5,190) .. controls (103.01,190) and (101,187.99) .. (101,185.5) -- cycle ;
\draw  [fill={rgb, 255:red, 0; green, 0; blue, 0 }  ,fill opacity=1 ] (141,185.5) .. controls (141,183.01) and (143.01,181) .. (145.5,181) .. controls (147.99,181) and (150,183.01) .. (150,185.5) .. controls (150,187.99) and (147.99,190) .. (145.5,190) .. controls (143.01,190) and (141,187.99) .. (141,185.5) -- cycle ;
\draw  [fill={rgb, 255:red, 0; green, 0; blue, 0 }  ,fill opacity=1 ] (181,185.5) .. controls (181,183.01) and (183.01,181) .. (185.5,181) .. controls (187.99,181) and (190,183.01) .. (190,185.5) .. controls (190,187.99) and (187.99,190) .. (185.5,190) .. controls (183.01,190) and (181,187.99) .. (181,185.5) -- cycle ;

\draw  [fill={rgb, 255:red, 0; green, 0; blue, 0 }  ,fill opacity=1 ] (101,105.5) .. controls (101,103.01) and (103.01,101) .. (105.5,101) .. controls (107.99,101) and (110,103.01) .. (110,105.5) .. controls (110,107.99) and (107.99,110) .. (105.5,110) .. controls (103.01,110) and (101,107.99) .. (101,105.5) -- cycle ;
\draw  [fill={rgb, 255:red, 0; green, 0; blue, 0 }  ,fill opacity=1 ] (141,105.5) .. controls (141,103.01) and (143.01,101) .. (145.5,101) .. controls (147.99,101) and (150,103.01) .. (150,105.5) .. controls (150,107.99) and (147.99,110) .. (145.5,110) .. controls (143.01,110) and (141,107.99) .. (141,105.5) -- cycle ;
\draw  [fill={rgb, 255:red, 0; green, 0; blue, 0 }  ,fill opacity=1 ] (181,105.5) .. controls (181,103.01) and (183.01,101) .. (185.5,101) .. controls (187.99,101) and (190,103.01) .. (190,105.5) .. controls (190,107.99) and (187.99,110) .. (185.5,110) .. controls (183.01,110) and (181,107.99) .. (181,105.5) -- cycle ;

\draw    (105.5,184) -- (145.33,184) ;
\draw [shift={(130.42,184)}, rotate = 180] [fill={rgb, 255:red, 0; green, 0; blue, 0 }  ][line width=0.08]  [draw opacity=0] (10.72,-5.15) -- (0,0) -- (10.72,5.15) -- (7.12,0) -- cycle    ;
\draw    (144.42,147.08) -- (144.42,186.92) ;
\draw [shift={(144.42,160.5)}, rotate = 90] [fill={rgb, 255:red, 0; green, 0; blue, 0 }  ][line width=0.08]  [draw opacity=0] (10.72,-5.15) -- (0,0) -- (10.72,5.15) -- (7.12,0) -- cycle    ;
\draw    (144.5,106) -- (144.5,145.83) ;
\draw [shift={(144.5,119.42)}, rotate = 90] [fill={rgb, 255:red, 0; green, 0; blue, 0 }  ][line width=0.08]  [draw opacity=0] (10.72,-5.15) -- (0,0) -- (10.72,5.15) -- (7.12,0) -- cycle    ;
\draw    (147.42,145.08) -- (147.42,184.92) ;
\draw [shift={(147.42,170)}, rotate = 270] [fill={rgb, 255:red, 0; green, 0; blue, 0 }  ][line width=0.08]  [draw opacity=0] (10.72,-5.15) -- (0,0) -- (10.72,5.15) -- (7.12,0) -- cycle    ;
\draw    (145.5,104) -- (185.33,104) ;
\draw [shift={(170.42,104)}, rotate = 180] [fill={rgb, 255:red, 0; green, 0; blue, 0 }  ][line width=0.08]  [draw opacity=0] (10.72,-5.15) -- (0,0) -- (10.72,5.15) -- (7.12,0) -- cycle    ;
\draw    (145.5,107) -- (185.33,107) ;
\draw [shift={(158.92,107)}, rotate = 0] [fill={rgb, 255:red, 0; green, 0; blue, 0 }  ][line width=0.08]  [draw opacity=0] (10.72,-5.15) -- (0,0) -- (10.72,5.15) -- (7.12,0) -- cycle    ;
\draw    (147.42,147.92) -- (147.42,108.08) ;
\draw [shift={(147.42,134.5)}, rotate = 270] [fill={rgb, 255:red, 0; green, 0; blue, 0 }  ][line width=0.08]  [draw opacity=0] (10.72,-5.15) -- (0,0) -- (10.72,5.15) -- (7.12,0) -- cycle    ;
\draw    (105.5,187) -- (145.33,187) ;
\draw [shift={(118.92,187)}, rotate = 0] [fill={rgb, 255:red, 0; green, 0; blue, 0 }  ][line width=0.08]  [draw opacity=0] (10.72,-5.15) -- (0,0) -- (10.72,5.15) -- (7.12,0) -- cycle    ;
\end{tikzpicture}
    \caption{The doubled loop connects the bottom left point and the top right point.}
    \label{fig:loop}
\end{figure}

At this point, the only remaining issue is that $(a_1,b_1)$ need not be nearest neighbors. We therefore express the factor $1-\what{\mr{SWAP}}_{a_1,b_1}$ as a sum of local projectors. To this end, we choose a self-avoiding nearest-neighbor path $\gamma=(\gamma_l)_{l=0}^L$ with $\gamma_0=a_1$ and $\gamma_L=b_1$ (see Fig.\,\ref{fig:loop}), and define the doubled loop $\widetilde{\gamma}$ by
\begin{equation}
    \widetilde{\gamma}_l=\begin{dcases}
        \gamma_l\,\,&(0\leq l\leq L)\\
        \gamma_{2L-l}\,\,&(L+1\leq l\leq 2L).
    \end{dcases}
\end{equation}
Then the transposition $t_{\widetilde{\gamma}_0,\widehat{\gamma}_L}$ can be written as a product of neighboring transpositions along the loop:
\begin{equation}
    t_{\wtil{\gamma}_0,\wtil{\gamma}_L}=t_{\wtil{\gamma}_0,\wtil{\gamma}_1}t_{\wtil{\gamma}_1,\wtil{\gamma}_2}\cdots t_{\wtil{\gamma}_{L-1},\wtil{\gamma}_L}t_{\wtil{\gamma}_{L-2},\wtil{\gamma}_{L-1}}\cdots t_{\wtil{\gamma}_0,\wtil{\gamma}_1}.
\end{equation}
Consequently, $1-t_{\wtil{\gamma}_0,\wtil{\gamma}_L}(\in\mb{C}[\mf{S}_N])$ admits the telescoping expression,
\begin{equation}
    \begin{split}
        1-t_{\gamma_0,\gamma_L}&=1-t_{\wtil{\gamma}_0,\wtil{\gamma}_1}+t_{\wtil{\gamma}_0,\wtil{\gamma}_1}-t_{\wtil{\gamma}_0,\wtil{\gamma}_1}t_{\wtil{\gamma}_1,\wtil{\gamma}_2}+\cdots\\
        &=\sum_{k=0}^{2L}\prod_{l=1}^kt_{\wtil{\gamma}_{l-1},\wtil{\gamma}_l}\big(1-t_{\wtil{\gamma}_l,\wtil{\gamma}_{l+1}}\big),
    \end{split}
\end{equation}
where $t_{\wtil{\gamma}_{-1},\wtil{\gamma}_0}\coloneqq1$. Under the representation on $\mc{H}$, each $1-t_{\widetilde{\gamma}_k,\widetilde{\gamma}_{k+1}}$ becomes $1-\what{\mr{SWAP}}_{\wtil{\gamma}_k,\wtil{\gamma}_{k+1}}$, i.e., a strictly local two-site projector on nearest neighbors, annihilating the symmetric part of wavefunction. This proves Lemma\,\ref{lemma:main} for $\mf{h}^s=\mf{h}$.

\textbf{Step 2: the general case $\mf{h}_x^s\subset\mf{h}_x$.}

Now let $\mf{h}_x^s\subset\mf{h}_x$. We introduce the on-site  projector $\hat{P}^s_x\in\mc{L}(\mf{h}_x)$ onto $\mf{h}^s_x$ and the associated global projector $\hat{P}^s\coloneqq\prod_{x\in\Lambda}\hat{P}_x^s$. The operator $\hat{O}$ trivially decomposes to 
\begin{equation}
    \hat{O}=\hat{O}\hat{P}^s+\hat{O}\big(1-\hat{P}^s\big).
\end{equation}
The second term trivially annihilates $\mr{Sym}^N(\mf{h}^s)$ since $1-\hat{P}^s$ annihilates any vector supported entirely in $\bigotimes_{x\in\Lambda}\mf{h}_x^s$. Therefore, $\hat{O}\hat{P}^s$ must individually annihilate $\mr{Sym}^N(\mf{h}^s)$. Applying Step 1 to the restricted Hilbert space $\bigotimes_{x\in\Lambda}\mf{h}^s_x$, we obtain a decomposition of $\hat{O}\hat{P}^s$ into terms containing strictly local projectors, i.e., $\hat{O}\hat{P}^s=\sum_{\expval{x,y}}\hat{o}^{(2)}_{[xy]}\hat{P}_{xy}$.

Finally, we expand the remaining factor $1-\hat{P}^s$ into a sum of local projectors. Indeed,
\begin{equation}\label{eq:hs projector}
    1-\hat{P}^s=1-\prod_{x\in\Lambda}\hat{P}_x^s=1-\prod_{x\in\Lambda}^N\Big(1-\big(1-\hat{P}_x^s\big)\Big),
\end{equation}
and the right-hand side is a (finite) linear combination of products of the strictly local projectors $1-\hat{P}^s_x$. This completes the proof of Lemma\,\ref{lemma:main}.
\end{proof}

Lemma\,\ref{lemma:main} immediately implies a structural decomposition for Hamiltonians preserving the ferromagnetic scar manifold. Suppose that $\mr{Sym}^N(\mf{h}^s)$ is invariant under $\hat{H}$. We can split $\hat{H}$ as,
\begin{equation}\label{eq:symmetrization decomp}
    \hat{H}=\big(\hat{H}-\mc{S}(\hat{H})\big)+\mc{S}(\hat{H})\eqqcolon\hat{H}_A+\hat{H}_Z.
\end{equation}
Invariance of $\mr{Sym}^N(\mf{h}^s)$ under $\hat{H}$ implies that $\hat{H}$ and $\mc{S}(\hat{H})$ act identically on this subspace. Equivalently, for any $\sigma\in\mf{S}_N$,
\begin{equation}
    \hat{H}\mr{Sym}^N(\mf{h}_s)=\hat{\sigma}\hat{H}\hat{\sigma}^{-1}\mr{Sym}^N(\mf{h}_s),
\end{equation}
so the ``asymmetric part'' $\hat{H}_A=\hat{H}-\mc{S}(\hat{H})$ annihilates $\mr{Sym}^N(\mf{h}^s)$. By Lemma\,\ref{lemma:main}, $\hat{H}_A$ therefore decomposes into terms containing local projectors that annihilate $\mr{Sym}^N(\mf{h}^s)$ locally. Moreover, by Eq.\,\eqref{eq:collective su(V)}, the symmetric part $\hat{H}_Z(=\mc{S}(\hat{H}))$ lies in the algebra generated by collective $\mf{su}(\mf{h}_s)$ operators; the simplest non-trivial term is a Zeeman-type on-site term.

This argument can be summarized as follows.
\begin{corollary}\label{cor:H decomp}
    If $\hat{H}$ leaves $\mr{Sym}^N(\mf{h}^s)$ invariant, it can be decomposed as,
    \begin{equation}
        \begin{split}
            &\hat{H}=\hat{H}_A+\hat{H}_Z\\
            &\hat{H}_A\mr{Sym}^N(\mf{h}^s)=0\\
            &\hat{H}_Z\ket{\psi_\lambda}=E_\lambda\ket{\psi_\lambda},
        \end{split}
    \end{equation}
    where $\ket{\psi_\lambda}\in\mr{Sym}^N(\mf{h}^s)$ is an eigenstate of a totally symmetric operator $\hat{H}_Z$. Furthermore, $\hat{H}_A$ is decomposed as a sum of terms containing local projectors,
    \begin{equation}
        \hat{H}_A=\sum_x\hat{h}_{[x]}^{(1)}\hat{P}_x+\sum_{\expval{x,y}}\hat{h}^{(2)}_{[xy]}\hat{P}_{xy}.
    \end{equation}
\end{corollary}
We again note that the operators $\hat{h}^{(1)}_{[x]}$ and $\hat{h}^{(2)}_{[xy]}$ can in principle be non-local at this point.

\section{Locality implies a Zeeman term}\label{sec:locality}
\subsection*{Locality}
Section\,\ref{sec:HA} (in particular Lemma\,\ref{lemma:main} or Cor.\,\ref{cor:H decomp}) required neither a spatial notion of locality nor any detailed structure of the exact eigenstates beyond permutation symmetry. In this section, we impose two additional assumptions:

\begin{itemize}
    \item[\textbf{1}.] the Hamiltonian is local, i.e., a sum of operators supported on connected subregions
    \item[\textbf{2}.] the totally symmetric weight basis $\{\ket*{\psi_{\vec{m}}^\Lambda}\}_{\vec{m}}\subset\mr{Sym}^N(\mf{h}^s)$ (Eq.\,\eqref{eq:Cartan basis state} or Eq.\,\eqref{eq:permutation expression}) consists of exact eigenstates 
\end{itemize}

Under these assumptions, we show that the only term that can act non-trivially within the scar manifold is necessarily a Zeeman term, i.e., a linear combination of on-site Cartan generators.

To adapt the representation-theoretic tools of Sec.\,\ref{sec:HA} to the local setting, let $\mf{S}_X\subset\mf{S}_N$ denote the subgroup of permutations acting non-trivially only on sites inside a subregion $X\subset\Lambda$. We write the corresponding Young symmetrizer (right action) as $\hat{c}^R_{T;X}$, which therefore projects the Hilbert space onto the irrep labeled by $T$ on the region $X$,
\begin{equation}
    \hat{c}^R_{T;X}\mc{H}\cong\mb{S}_{T;X}\mf{h}\otimes\mc{H}_{X^c},
\end{equation}
where $\mc{H}_{X^c}=\bigotimes_{x\not\in X}\mf{h}_x$ is the Hilbert space of the complement $X^c$.

Finally, we denote the projector onto $\mf{h}^s$ on a subregion $X$ by
\begin{equation}
    \hat{P}_X^s\coloneqq\bigotimes_{x\in X}\hat{P}^s_x.
\end{equation}

\subsection*{Local decomposition}
We begin with a lemma on operators on $\bigotimes_{x\in\Lambda}\mf{h}_x^s$.

\begin{lemma}\label{lemma:su(V)}
    We regard $\mf{h}^s$ as the fundamental representation of $\mf{su}(\mf{h}^s)$. 
    \begin{itemize}
    \item[(1)] Any operator $\hat{O}\in\mc{L}(\bigotimes_{x=1}^N\mf{h}_x^s)$ commuting with all collective Cartan generators $\hat{H}_{\mr{tot}}^i=\sum_{x=1}^N\hat{H}_x^i$ can be written as a polynomial in the on-site Cartan generators and permutation operators,
    \begin{equation}
        \hat{O}\in\left\langle\hat{H}_x^i;\,1\leq i\leq d_s-1,\,1\leq x\leq N\right\rangle\otimes\Big\langle\hat{\sigma};\sigma\in\mf{S}_N\Big\rangle.
    \end{equation}
    \item[(2)] If $\mf{h}^s\cong\mb{C}^2$, then any operator $\hat{O}\in\mc{L}(\bigotimes_{x=1}^N\mf{h}_x^s)$ belongs to
    \begin{equation}
        \hat{O}\in\Big\langle\hat{\sigma}_x^3,\,\hat{\sigma}_x^+;1\leq x\leq N\Big\rangle\otimes\Big\langle\hat{\sigma};\sigma\in\mf{S}_N\Big\rangle\cup\Big\langle\hat{\sigma}_x^3,\,\hat{\sigma}_x^-;1\leq x\leq N\Big\rangle\otimes\Big\langle\hat{\sigma};\sigma\in\mf{S}_N\Big\rangle,
    \end{equation}
    where $\hat{\sigma}^3$ and $\hat{\sigma}^\pm$ are the usual Pauli-$z$ operator and ladder operator, respectively.
    \end{itemize}
\end{lemma}
We defer the proof to Appendix\,\ref{proof:su(V)}.

We can now state the main result of this section.
\begin{theorem}\label{thm:local decomp}
    Suppose that $\hat{H}$ is local,
    \begin{equation}
        \hat{H}=\sum_{X\subset\Lambda}\hat{h}_X.
    \end{equation}
    If every totally symmetric weight basis state $\ket*{\psi^\Lambda_{\vec{m}}}\in\mr{Sym}^N(\mf{h}^s)$ is an exact eigenstate of $\hat{H}$, then $\hat{H}$ admits a decomposition
    \begin{equation}
        \hat{H}=\sum_{x\in\Lambda}\hat{h}^{(1)}_{[x]}\hat{P}_x+\sum_{\expval{x,y}}\hat{h}^{(2)}_{[xy]}\hat{P}_{xy}+\sum_{x\in\Lambda}\hat{h}_x,
    \end{equation}
    where the first two sums are annihilator terms built from local projectors $\hat{P}_x$ and $\hat{P}_{xy}$ that locally annihilate $\mr{Sym}^N(\mf{h}^s)$, and the last term $\sum_{x\in\Lambda}\hat{h}_x$ is a Zeeman term, i.e., a linear combination of on-site Cartan generators $\{\hat{H}_x^i\}_i$.
\end{theorem}
\begin{proof}
    We first fix a connected region $X\subset\Lambda$. We start by splitting each local term $\hat{h}_X$ according to whether it preserves the local target space $\bigotimes_{x\in X}\mf{h}_x^s$,
    \begin{equation}\label{eq:hx decomp}
        \hat{h}_X=\hat{h}_X\big(1-\hat{P}_X^s\big)+\big(1-\hat{P}_X^s\big)\hat{h}_X\hat{P}_X^s+\hat{P}_X^s\hat{h}_X\hat{P}_X^s.
    \end{equation}
    The first term $\hat{h}_X(1-\hat{P}_X^s)$ annihilates $\mr{Sym}^N(\mf{h}^s)$ trivially. Moreover, using Eq.\,\eqref{eq:hs projector}, it can be expanded as a sum of (strictly) local projectors. 
    
    The second term pushes states in $\mr{Sym}^N(\mf{h}^s)\subset\bigotimes_{x\in\Lambda}\mf{h}_x^s$ into the orthogonal complement: 
    \begin{equation}
        \big(1-\hat{P}_X^s\big)\hat{h}_X\hat{P}_X^s\mr{Sym}^N(\mf{h}^s)\subset\left(\bigotimes_{x\in\Lambda}\mf{h}_x^s\right)^\perp.
    \end{equation}
    Since the third term $\hat{P}_X^s\hat{h}_X\hat{P}_X^s$ acts within $\bigotimes_{x\in\Lambda}\mf{h}^s_x$, invariance of $\mr{Sym}^N(\mf{h}^s)$ forces cancellation of the leakage contributions (at least) when summed over $X$,
    \begin{equation}\label{eq:(1-P)hP annihilation}
        \sum_{X\subset\Lambda}\big(1-\hat{P}_X^s\big)\hat{h}_X\hat{P}_X^s\mr{Sym}^N(\mf{h}^s)=0.
    \end{equation}
    By Lemma\,\ref{lemma:main}, this annihilating operator can be written as a sum of terms containing local projectors that annihilate $\mr{Sym}^N(\mf{h}^s)$.

    It remains to analyze the last term in Eq.\,\eqref{eq:hx decomp}. Decompose $\hat{P}_X^s\hat{h}_X\hat{P}_X^s$ using the local Young symmetrizer resolution over $\mf{S}_X$:
    \begin{equation}\label{eq:PhP decomp}
        \hat{P}_X^s\hat{h}_X\hat{P}_X^s=\hat{P}_X^s\hat{h}_X\hat{P}_X^s\left(\sum_T\hat{c}^R_{T;X}\right)=\hat{P}_X^s\hat{h}_X\hat{P}_X^s\hat{c}^R_{(|X|);X}+\sum_{T\not=(|X|)}\hat{P}_X^s\hat{h}_X\hat{P}_X^s\hat{c}^R_{T;X}.
    \end{equation}
    For $T\not=(|X|)$, the projector $\hat{c}^R_{T;X}$ contains anti-symmetrizing components on $X$, hence it annihilates the totally symmetric sector (this is also trivial from the Littlewood-Richardson coefficient). Therefore, the second term in Eq.\,\eqref{eq:PhP decomp} locally annihilates $\mr{Sym}^N(\mf{h}^s)$. 
    
    Next, we split the symmetric component into weight-preserving and weight-non-preserving parts:
    \begin{equation}
        \hat{P}_X^s\hat{h}_X\hat{P}_X^s\hat{c}^R_{(|X|);X}=\hat{h}_X^{\mr{cons}}+\hat{h}_X^{\mr{non}}.
    \end{equation}
    By definition, $\hat{h}_X^{\mr{non}}$ changes the total weight (i.e., the eigenvalues of $\{\hat{H}_{\mr{tot}}^i\}_i$). Thus, when acting on $\ket*{\psi_{\vec{m}}^\Lambda}$, it either annihilates the state or produces a component lying in a different $\{\hat{H}_{\mr{tot}}^i\}_i$-sector, orthogonal to the original one. Since the weight-preserving terms $\{\hat{h}_X^{\mr{cons}}\}_X$ never generate such components, the condition that every $\ket*{\psi^\Lambda_{\vec{m}}}$ is an exact eigenstate of $\hat{H}$ forces the total weight-changing contribution to vanish on $\mr{Sym}^N(\mf{h}^s)$:
    \begin{equation}
        \sum_{X\subset\Lambda}\hat{h}_X^{\mr{non}}\mr{Sym}^N(\mf{h}^s)=0.
    \end{equation}
    Again, Lemma\,\ref{lemma:main} implies that this contribution is an annihilator built from local projectors.

    We now focus on the weight-preserving part. By Lemma\,\ref{lemma:su(V)}, each $\hat{h}_X^{\mr{cons}}$ can be expressed using local Cartan generators and permutations supported on $X$,
    \begin{equation}
        \hat{h}^{\mr{cons}}_X\in\left\langle\hat{H}^i_x;\,1\leq i\leq d_s-1,\,x\in X\right\rangle\otimes\Big\langle\hat{\sigma};\,\sigma\in\mf{S}_X\Big\rangle.
    \end{equation}
    Because $\ket*{\psi^\Lambda_{\vec{m}}}$ is totally symmetric, $\hat{\sigma}\ket*{\psi^\Lambda_{\vec{m}}}=\ket*{\psi^\Lambda_{\vec{m}}}$ for $\forall\vec{m}$ and $\forall\sigma\in\mf{S}_N$. Consequently, the action of $\hat{h}_X^{\mr{cons}}$ on $\ket*{\psi^\Lambda_{\vec{m}}}$ is ``diagonal'' in the sense that there exists an operator-valued function $F_X$ such that 
    \begin{equation}\label{eq:hX cons action}\begin{split}
        \hat{h}_X^{\mr{cons}}\ket*{\psi^\Lambda_{\vec{m}}}=F_X\big[\{\hat{H}_x^i\}_{x\in X}\big]\ket*{\psi^\Lambda_{\vec{m}}}.
        \end{split}
    \end{equation}
    Using the explicit permutation-symmetrized form of $\ket*{\psi^\Lambda_{\vec{m}}}$ (Eq.\,\eqref{eq:permutation expression}), Eq.\,\eqref{eq:hX cons action} becomes
    \begin{equation}\begin{split}\label{eq:h_X cartan}
        \hat{h}_X^{\mr{cons}}\ket*{\psi^\Lambda_{\vec{m}}}&=\frac{1}{\mc{N}_{\vec{m}}}\sum_{\sigma\in\mf{S}_N}F_X\big[\{\hat{H}_x^i\}_{x\in X}\big]\bigotimes_{x\in\Lambda}\ket*{\vec{m}_{\sigma(x)}}_x\\
        &=\frac{1}{\mc{N}_{\vec{m}}}\sum_{\sigma\in\mf{S}_N}F_X\big[\{m_{\sigma(x)}^i\}_{x\in X}\big]\bigotimes_{x\in\Lambda}\ket*{\vec{m}_{\sigma(x)}}_x,
    \end{split}
    \end{equation}
    where $\{\vec{m}_x\}_{x\in\Lambda}$ is a fixed reference configuration whose orbit under permutations generates $\ket*{\psi^\Lambda_{\vec{m}}}$. 
    
    Summing Eq.\,\eqref{eq:h_X cartan} over all local regions $X$ shows that the action of the full Hamiltonian on $\ket*{\psi^\Lambda_{\vec{m}}}$ is equivalent to that of a diagonal operator $\hat{F}$,
    \begin{equation}
        \begin{split}
            \hat{H}\ket*{\psi^\Lambda_{\vec{m}}}&=\frac{1}{\mc{N}_{\vec{m}}}\sum_{\sigma\in\mf{S}_N}\hat{H}\bigotimes_{x\in\Lambda}\ket*{\vec{m}_{\sigma(x)}}_x\\
            &=\frac{1}{\mc{N}_{\vec{m}}}\sum_{\sigma\in\mf{S}_N}\sum_{X\subset\Lambda}F_X\big[\{m^i_{\sigma(x)}\}_{x\in X}\big]\bigotimes_{x\in\Lambda}\ket*{\vec{m}_{\sigma(x)}}_x\\
            &\equiv\hat{F}\ket*{\psi^\Lambda_{\vec{m}}}\\
            \hat{F}&\coloneqq\sum_{X\subset\Lambda}F_X\big[\{m_x^i\}_{x\in X}\big]\bigotimes_{x\in X}\dyad*{\vec{m}_x}_x.
        \end{split}
    \end{equation}
    Since $\ket*{\psi^\Lambda_{\vec{m}}}$ is an equal superposition over all permutations of the reference configuration $\bigotimes_{x\in\Lambda}\ket*{\vec{m}_x}_x$, and $\hat{F}$ is diagonal, we have  
    \begin{equation}\begin{split}
        &\begin{dcases}
            \hat{F}\ket*{\psi^\Lambda_{\vec{m}}}=E\ket*{\psi^\Lambda_{\vec{m}}}\\
            \ket*{\psi^\Lambda_{\vec{m}}}=\frac{1}{\mc{N}'_{\vec{m}}}\sum_{\sigma\in\mf{S}_N}\bigotimes_{x\in\Lambda}\ket*{\vec{m}_{\sigma(x)}}_x\,\,(\text{Eq.\,\eqref{eq:permutation expression}})\\
            \hat{F}\bigotimes_{x\in\Lambda}\ket*{\vec{m}_x}_x=\left(\sum_{X\subset\Lambda}F_X\big[\{m_x^i\}_{x\in X}\big]\right)\bigotimes_{x\in\Lambda}\ket*{\vec{m}_x}_x
        \end{dcases}\\
        &\rightarrow\sum_{X\subset\Lambda}F_X[\{m_x^i\}_{x\in X}]=\sum_{X\subset\Lambda}F_X[\{m_{\sigma(x)}^i\}_{x\in X}]\,\,\text{for }\forall\sigma\in\mf{S}_N.
    \end{split}
    \end{equation}
    The condition in the last line is true only if $\hat{F}$ is totally symmetric under site permutations:
    \begin{equation}\begin{split}
        \mc{S}(\hat{F})=\hat{F}\in\mr{Comm}(\mf{S}_N)
        \end{split}
    \end{equation}
    By Eq.\,\eqref{eq:collective su(V)}, any element of $\mr{Comm}(\mf{S}_N)$ is a polynomial in the collective $\mf{su}(\mf{h}^s)$ generators. Finally, we note that $\hat{F}$ inherits locality from $\hat{H}=\sum_{X\subset\Lambda}\hat{h}_X$. Among permutation-invariant polynomials, the only forms compatible with strict locality (up to an additive constant) are one-body collective Cartan terms (any polynomial of degree$\geq2$ is necessarily all-to-all), i.e., a Zeeman term,
    \begin{equation}
        \hat{F}=\sum_{i=1}^{d_s-1}h_i\sum_{x\in\Lambda}\hat{H}_x^i.
    \end{equation}
    This identifies the unique non-annihilating contribution within $\mr{Sym}^N(\mf{h}^s)$ with a Zeeman term and completes the proof.
\end{proof}

\section{Local annihilators and Dzyaloshinskii--Moriya-type terms}\label{sec:quasi-local decomp}
In the proof of Lemma\,\ref{lemma:main}, an annihilator $\hat{O}$ of $\mr{Sym}^N(\mf{h}^s)$ is written as a sum of terms containing strictly local projectors (Eq.\,\eqref{eq:lemma main}). However, the accompanying operator factors $\hat{o}_x^{(1)}$ and $\hat{o}_{xy}^{(2)}$ can be highly non-local in general, simply because Lemma\,\ref{lemma:main} does not assume any spatial structure of $\hat{O}$. By contrast, the proof of Thm.\,\ref{thm:local decomp} shows that the annihilating operators arising from a local Hamiltonian naturally appear as sums of local terms, namely,
\begin{equation}\label{eq:annihilator summary}
    \begin{split}
        \sum_{X\subset\Lambda}\hat{h}_X(1-\hat{P}_X^s),\,\sum_{X\subset\Lambda}\big(1-\hat{P}_X^s\big)\hat{h}_X\hat{P}_X^s,\,\sum_{X\subset\Lambda}\sum_{T\not=(|X|)}\hat{P}_X^s\hat{h}_X\hat{P}_X^s\hat{c}_{T;X}^R,\,\sum_{X\subset\Lambda}\hat{h}_X^{\mr{non}},
    \end{split}
\end{equation}
each of which annihilates $\mr{Sym}^N(\mf{h}^s)$. 

These terms split into two qualitatively different classes. The first and the third sums in Eq.\,\eqref{eq:annihilator summary} annihilate $\mr{Sym}^N(\mf{h}^s)$ in a straightforward local manner: the projectors $(1-\hat{P}_X^s)$ and $\hat{c}_{T;X}^R$ already annihilate $\mr{Sym}^N(\mf{h}^s)$ by themselves. The second and the fourth sums are subtler: each local term may act non-trivially on $\mr{Sym}^N(\mf{h}^s)$, and the annihilation is enforced only by the assumption that $\mr{Sym}^N(\mf{h}^s)$ is invariant under the full Hamiltonian, i.e., contributions into the orthogonal complement must cancel in the total sum. Consequently, while Lemma\,\ref{lemma:main} guarantees that these operators admit decompositions into terms containing local projectors, the corresponding prefactors may be non-local. 

In this section, we show that the second term $\sum_{X\subset\Lambda}(1-\hat{P}_X^s)\hat{h}_X\hat{P}_X^s$ in Eq.\,\eqref{eq:annihilator summary} in fact admits a decomposition into a sum of strictly local annihilators, and we further constrain the possible structure of the fourth term $\sum_{X\subset\Lambda}\hat{h}_X^{\mr{non}}$ for $\mf{h}^2\cong\mb{C}^2$, where Dzyaloshinskii-Moriya (DM)-type interactions appear naturally.

To this end, we first extend Lemma\,\ref{lemma:main} to the local setting.
\begin{lemma}\label{lemma:local HA}
    If an operator $\hat{O}_X$ acts non-trivially only on a subregion $X\subset\Lambda$ and annihilates $\mr{Sym}^N(\mf{h}^s)$, then it can be expressed as a sum of local terms supported within $X$, which contain local projectors,
    \begin{equation}\label{eq:annihilator X decomp}
        \hat{O}_X=\sum_{x\in X}\hat{o}_{[x]}^{(1)}\hat{P}_x+\sum_{\langle x,y\rangle\subset X}\hat{o}_{[xy]}^{(2)}\hat{P}_{xy},
    \end{equation}
    where $\langle x,y\rangle\subset X$ denotes nearest neighbors within $X$, and $\hat{P}_x$ and $\hat{P}_{xy}$ are local projectors acting non-trivially on $x$ and $(x,y)$, respectively, each annihilating $\mr{Sym}^N(\mf{h}^s)$. Here, $\hat{o}^{(1)}_{[x]}$ and $\hat{o}^{(2)}_{[xy]}$ are also local and act non-trivially only within $X$.
\end{lemma}
We leave the proof in Appendix\,\ref{proof:local HA}.


\begin{prop}\label{prop:(1-P)}
    The operator $\sum_{X\subset\Lambda}(1-\hat{P}_X^s)\hat{h}_X\hat{P}_X^s$, which annihilates $\mr{Sym}^N(\mf{h}^s)$, can be decomposed as a sum of strictly local operators, each of which annihilates $\mr{Sym}^N(\mf{h}^s)$.
\end{prop}
\begin{proof}
    We define the operator $\hat{Q}_X^s$ as,
    \begin{equation}
        \hat{Q}_X^s\coloneqq\prod_{x\in X}\big(1-\hat{P}_x^s\big).
    \end{equation}
    Expanding $1-\hat{P}_X^s$ by inclusion-exclusion on $X$ yields the decomposition,  
    \begin{equation}\label{eq:decomp prop 7.1}
        \big(1-\hat{P}_X^s\big)\hat{h}_X\hat{P}_X^s=\sum_{Y\subset X}\hat{Q}_Y^s\hat{P}_{X\setminus Y}^s\hat{h}_X\hat{P}_X^s.
    \end{equation}
    Consequently, the full second term in Eq.\,\eqref{eq:annihilator summary} can be reorganized as, 
    \begin{equation}
        \sum_{X\subset\Lambda}(1-\hat{P}_X^s)\hat{h}_X\hat{P}_X^s=\sum_{X\subset\Lambda}\hat{Q}_X^s\hat{O}_X,
    \end{equation}
    for suitable local operator $\hat{O}_X$. This reorganization preserves locality: $\hat{O}_X$ is supported on a region no smaller than $X$, and, crucially, outside $X$ the state remains inside $\bigotimes_{x\in X^c}\mf{h}_x^s$. More precisely,
    \begin{equation}\label{eq:Q_X^s action}
        \hat{Q}_X^s\hat{O}_X\mr{Sym}^N(\mf{h}^s)\subset\bigotimes_{x\in X}\big(\mf{h}^s_x\big)^\perp\bigotimes_{y\in X^c}\mf{h}^s_y.
    \end{equation}
    We now left-multiply Eq.\,\eqref{eq:decomp prop 7.1} by $\hat{Q}_Y^s$ for an arbitrary $Y\subset\Lambda$. Since the left-hand side annihilates $\mr{Sym}^N(\mf{h}^s)$, we have,
    \begin{equation}
        \begin{split}\label{eq:trivial Q_X^s identity}
        &\hat{Q}_Y^s\sum_{X\subset\Lambda}\hat{Q}_X^s\hat{O}_X\mr{Sym}^N(\mf{h}^s)=0.
        \end{split}
    \end{equation}
    Using Eq.\,\eqref{eq:Q_X^s action}, one obtains the ``selection rule'',
     \begin{equation}\label{eq:Q_X^s identity}
         \hat{Q}_Y^s\hat{Q}_X^s\hat{O}_X\mr{Sym}^N(\mf{h}^s)=\begin{dcases}
             \hat{Q}_X^s\hat{O}_X\mr{Sym}^N(\mf{h}^s)&(Y\subset X)\\
             0&(\text{else})
         \end{dcases},
     \end{equation}
     and therefore Eq.\,\eqref{eq:trivial Q_X^s identity} implies the local annihilation constraint,
     \begin{equation}
         \sum_{Y\subset X}\hat{Q}_X^s\hat{O}_X\mr{Sym}^N(\mf{h}^s)=0.
     \end{equation}
     By choosing $Y$ appropriately, one can isolate individual contributions $\hat{Q}^s_Y\hat{O}_Y$ for large enough $Y$. Lemma\,\ref{lemma:local HA} then implies that each such term admits a decomposition into strictly local operators containing projectors that annihilate $\mr{Sym}^N(\mf{h}^s)$. Iterating this argument over all supports completes the proof.
\end{proof}

As a consequence of Prop.\,\ref{prop:(1-P)}, the only potentially problematic contribution in Eq.\,\eqref{eq:annihilator summary} is $\sum_{X\subset\Lambda}\hat{h}_X^{\mr{non}}$, i.e., the sum of the weight-non-preserving terms in $\mc{L}(\bigotimes_{x\in\Lambda}\mf{h}_x^s)$. Therefore one can show that the absence of the non-preserving part implies the existence of strictly local decompositions of Hamiltonian.
\begin{corollary}
    Let $\hat{H}=\sum_{X\subset\Lambda}\hat{h}_X\in\mc{L}(\bigotimes_{x\in\Lambda}\mf{h}_x)$ be local. If every totally symmetric weight basis $\ket*{\psi^\Lambda_{\vec{m}}}\in\mr{Sym}^N(\mf{h}^s)$ is an exact eigenstate and each $\hat{P}_X^s\hat{h}_X\hat{P}_X^s$ is weight preserving, i.e., $[\hat{P}_X^s\hat{h}_X\hat{P}_X^s,\hat{H}_X^i]=0$ for $\forall i$, then $\hat{H}$ admits the decomposition,
    \begin{equation}
        \hat{H}=\sum_{x\in\Lambda}\hat{h}_{[x]}^{(1)}\hat{P}_x^{(1)}+\sum_{\langle x,y\rangle}\hat{h}^{(2)}_{[xy]}\hat{P}_{xy}^{(2)}+\sum_{x\in\Lambda}\hat{h}_x,
    \end{equation}
    where $\hat{P}_x^{(1)}$ and $\hat{P}_{xy}^{(2)}$ are local projectors annihilating $\mr{Sym}^N(\mf{h}^s)$, and $\hat{h}_{[x]}^{(1)}$ and $\hat{h}_{[xy]}^{(2)}$ are strictly local. The energy spectrum of $\ket*{\psi_{\vec{m}}}$ is determined by the last Zeeman term $\sum_{x\in\Lambda}\hat{h}_x$.
\end{corollary}

We cannot prove in full generality that the fourth term in Eq.\,\eqref{eq:annihilator summary} decomposes into a sum of locally annihilating operators. Nevertheless, for $\mf{h}^s\cong\mb{C}^2$, we can substantially 
restrict the allowed form of $\hat{h}_X^{\mr{non}}$, in a way that naturally includes DM-type interactions. The key ingredient is the following statement on linear independence.
\begin{lemma}\label{lemma:su(2) linearity}
    For the $l=2j+1$ dimensional representation of $SU(2)$, we consider an operator of the form,
    \begin{equation}
        \sum_{a,b}\xi_{a,b}\big(\hat{J}^3\big)^a\big(\hat{J}^+\big)^b+\sum_{c,d}\chi_{c,d}\big(\hat{J}^3\big)^c\big(\hat{J}^+\big)^d,\,\,a,b,c,d\in\mb{N},
    \end{equation}
    where $\hat{J}^3$ is the Cartan subalgebra ($z$-component of spin), and $\hat{J}^\pm$ is the ladder operator. If $a+b<l$ and $c+d<l$ for all appearing terms, and the operator is identically zero, then all coefficients $\xi_{a,b}$ and $\chi_{c,d}$ vanish.
\end{lemma}
We leave the proof in Appendix\,\ref{proof:su(2) linearity}.
\begin{prop}\label{prop:hnon C2 constraint}
    Suppose $\mf{h}_x^s\cong\mb{C}^2$. Then there exist local operators $\{\hat{h}_\mu\}_\mu\subset\expval{\hat{\sigma}_x^3,\hat{\sigma}^\pm_x;x\in\Lambda}$ and corresponding sets of permutations $\{\mf{R}_\mu\}_\mu\subset2^{\mf{S}_N}$ such that $\sum_{X\subset\Lambda}\hat{h}_X^{\mr{non}}$ admits the decomposition
    \begin{equation}
        \sum_{X\subset\Lambda}\hat{h}_X^{\mr{non}}=\sum_\mu\sum_{\sigma\in\mf{R}_\mu}c_\sigma\hat{\sigma}\hat{h}_\mu\hat{\sigma}^{-1}+\sum_{\langle x,y\rangle}\hat{h}_{[xy]}^{(2)}\hat{P}^{(2)}_{xy},
    \end{equation}
    where $\hat{h}_{[xy]}^{(2)}$ is a local term and $\hat{P}_{xy}^{(2)}$ is a local projector annihilating $\mr{Sym}^N(\mf{h}^s)$. The coefficient $\{c_\sigma\}_{\sigma\in\mf{R}}$ satisfies
    \begin{equation}\label{eq:hnon cond}
        \sum_{\sigma\in\mf{R}}c_\sigma=0.
    \end{equation}
    
\end{prop}
\begin{remark}
    The DM interaction $\sum_{x=1}^L\hat{\sigma}_x^3\hat{\sigma}_{x+1}^1-\hat{\sigma}_x^1\hat{\sigma}_{x+1}^3$ trivially satisfies Prop.\,\ref{prop:hnon C2 constraint}. 
\end{remark}
\begin{proof}
    We define the operator $\hat{\sigma}^\pm_X$ as,
    \begin{equation}
        \hat{\sigma}_X^\pm\coloneqq\prod_{x\in X}\hat{\sigma}_x^\pm.
    \end{equation}
    By Lemma\,\ref{lemma:su(V)} (2), each $\hat{h}_X^{\mr{non}}$ lies in a sum of the following two algebras,
    \begin{equation}
        \hat{h}_X^{\mr{non}}\in\Big\langle\hat{\sigma}_x^3,\hat{\sigma}_x^+;\,x\in X\Big\rangle\otimes\Big\langle\hat{\sigma};\,\sigma\in\mf{S}_{X}\Big\rangle\oplus\Big\langle\hat{\sigma}_x^3,\hat{\sigma}_x^-;\,x\in X\Big\rangle\otimes\Big\langle\hat{\sigma};\,\sigma\in\mf{S}_X\Big\rangle.
    \end{equation} 
    Accordingly, and analogously to Prop.\,\ref{prop:(1-P)}, we may rewrite the full sum as,
    \begin{equation}\begin{split}
        \sum_{X\subset\Lambda}\hat{h}_X^{\mr{non}}&=\sum_{X\subset\Lambda}\hat{\sigma}_{X}^+\hat{O}_{X}^++\sum_{X\subset\Lambda}\hat{\sigma}_{X}^-\hat{O}_{X}^-,
        \end{split}
    \end{equation}
    where each $X$ is taken to be distinct, and the decomposition preserves locality. Acting on a totally symmetric state $\ket*{\psi^\Lambda_m}$, permutations contained in $\hat{O}_X^\pm$ act trivially, so we may replace $\hat{O}^\pm_X$ by operators $\hat{O}'^\pm_X$ built only from $\hat{\sigma}^3$,
    \begin{equation}
        \sum_{X\subset\Lambda}\hat{h}_X^{\mr{non}}\ket*{\psi^\Lambda_m}=\left(\sum_{X\subset\Lambda}\hat{\sigma}_{X}^+\hat{O}'^+_{X}+\sum_{X\subset\Lambda}\hat{\sigma}^-_{X}\hat{O}'^-_{X}\right)\ket*{\psi^\Lambda_m}=0.
    \end{equation}
    Using again that $\ket*{\psi^\Lambda_m}$ is totally symmetric, we can symmetrize each term as
    \begin{equation}\begin{split}
        &\left(\sum_{X\subset\Lambda}\hat{\sigma}_{X}^+\hat{O}'^+_{X}+\sum_{X\subset\Lambda}\hat{\sigma}^-_{X}\hat{O}'^-_{X}\right)\ket*{\psi^\Lambda_m}    =\left(\sum_{X\subset\Lambda}\mc{S}\big(\hat{\sigma}_{X}^+\hat{O}'^+_{X}\big)+\sum_{X\subset\Lambda}\mc{S}\big(\hat{\sigma}^-_{X}\hat{O}'^-_{X}\big)\right)\ket*{\psi^\Lambda_m}=0.
        \end{split}
    \end{equation}
    By Eq.\,\eqref{eq:collective su(V)}, each symmetrized term $\mc{S}(\hat{\sigma}_{X}^\pm\hat{O}'^\pm_{X})$ is a collective $\mf{su}(\mf{h}^s)$ operator. Moreover, the resulting collective operator depends only on $|X|$ and on the number of $\hat{\sigma}^3$ factors appearing in $\hat{O}'^\pm_X$. We thus group the local terms $\hat{\sigma}_X^\pm\hat{O}'^\pm_X$ into classes $A_\pm^{(a,b)}$ labeled by the pair $(a,b)$, where $a$ counts the number of $\hat{\sigma}^\pm$ operators and $b$ counts the number of $\hat{\sigma}^3$ operators. Every element of $A^{(a,b)}_\pm$ is mapped to the same collective operator under $\mc{S}$, and Lemma\,\ref{lemma:su(2) linearity} then implies the constraint,
    \begin{equation}
        \sum_{\hat{h}\in\mc{A}_{(a,b)}}\mc{S}\big(\hat{h}\big)\mr{Sym}^N(\mf{h}^s)=0.
    \end{equation}
    Finally, since $\mc{S}$ is linear (additive) as a map on operators, Eq.\,\eqref{eq:hnon cond} follows.
\end{proof}


\section{Discussions}\label{sec:discussion}
In this paper, we established a structural necessity theorem for a broad class of QMBS Hamiltonians. Specifically, we proved that if a local Hamiltonian admits the full set of totally symmetric weight-basis states as exact eigenstates, then its operator structure is essentially fixed: it must decompose into an annihilator built from terms containing local projectors that locally annihilate the symmetric sector, together with a Zeeman term that generates an equidistant energy spectrum within the scar manifold. Our argument combines elementary representation theory of the symmetric group with the bicommutant theorem, which identifies the totally symmetric sector as a distinguished irreducible block and strongly constrains the operators that can act non-trivially within it

Many QMBS toy models, up to local invertible transformations, realize scar states of precisely the ``ferromagnetic/totally symmetric'' form. Our result therefore provides a unified explanation for a recurring empirical pattern. We have observed that such Hamiltonians almost inevitably admit SM-like decompositions, and our theorem upgrades this (empirical) observation into an exhaustiveness statement: for this class of scars, an appropriately generalized SM structure is not merely a convenient construction, but is actually necessary.

We expect this structural viewpoint to serve as a useful starting point for further work, both mathematically and physically. Below we highlight a few concrete directions that appear particularly promising.

On the mathematical side, an important open problem concerns the locality of annihilators in the weight-non-preserving sector, i.e., whether the operator $\sum_{X\subset\Lambda}\hat{h}_X^{\mr{non}}$ in Eq.\,\eqref{eq:annihilator summary} necessarily admits a genuinely local decomposition. Proposition\,\ref{prop:hnon C2 constraint} shows that, in the case $\mf{h}^s\cong\mb{C}^2$, a representative interaction in this sector is the DM term $\sum_{x=1}^L\hat{\sigma}_x^3\hat{\sigma}_{x+1}^1-\hat{\sigma}_x^1\hat{\sigma}_{x+1}^3$, which nevertheless admits a strictly local decomposition. This raises a natural question: does $\sum_{X\subset\Lambda}\hat{h}_X^{\mr{non}}$ always admit such a local decomposition, or can one construct a genuinely non-trivial example for which any expansion necessarily involves non-local terms?

A second direction is to extend the analysis from the full symmetric sector $\mr{Sym}^N(\mf{h}^s)$ to proper subspaces thereof. Simple examples include restricted families such as the ``even Dicke states'' $\big(\hat{J}^+\big)^{2n}\bigotimes_{x\in\Lambda}\ket*{\downarrow}$. For such subspaces, one may hope for analogues of Lemma\,\ref{lemma:main} and Thm.\,\ref{thm:local decomp}, potentially with modified local projectors reflecting the reduced representation content. This extension is particularly relevant for models such as the AKLT chain, where the scar states exist but the scar manifold is not the full symmetric power (see Ref.\,\cite{KeitaOmiyaPhD} for further discussion). Nevertheless, we note in passing that our theorem can be easily extended to the case where a local Hilbert space has a higher spin ($S\geq1$) and the scar manifold is identified with the highest weight space (with respect to the total spin), as long as the local space can be constructed by a totally symmetric combination of fundamental representations.

On the physical side, comparatively little is established about the universal consequences of Hamiltonians constrained to this form. 
At present, the most robust dynamical diagnostic of ferromagnetic scars is coherent revivals, typically interpreted as collective large-spin precession\,\cite{Choi2019}. While several works have suggested anomalous diffusion\,\cite{Ljubotina2023}, possible dynamical signatures\,\cite{PhysRevLett.130.250402,PhysRevB.110.144302,pizzi2024,b83y-h128}, even new nonequilibrium phases\,\cite{omiya2025}, establishing QMBS as a genuine physical phenomenon, beyond a mere mathematical curiosity, likely requires the development of sharper and more model-independent probes. In particular, can one identify transport or hydrodynamic phenomena tied specifically to ferromagnetic scar manifolds, in a way that cleanly distinguishes them from fine-tuned revivals?

Finally, a definitive characterization of ferromagnetic scars is closely tied to the stability of the projector--annihilator structure under generic local perturbations. Many-body eigenstates can be extremely sensitive to even weak perturbation, which may strongly deform exact scar eigenstates and blur their spectral signatures. Nevertheless, the notion of ``approximate'' QMBS is often invoked in experimentally motivated settings including the PXP model\cite{Turner2018prb,Hudomal2020,Kolb_2023}. It would therefore be valuable to develop a systematic framework for approximate scars: do they correspond to approximately annihilating local projectors, and how does the degree of approximation quantitatively control revival lifetimes and other dynamical observables?

\section*{Note Added}
A closely related work by Nicholas O'Dea\,\textit{et.al.} discusses energy spectra of similar $S=1/2$ generalized towers built from certain creation operators that go beyond simple totally symmetric wavefunctions\,\cite{NickLeiMot}. The two manuscripts were posted in a coordinated manner.  

\section*{Acknowledgement}
The author is indebted to Nicholas O'Dea for their careful reading of the manuscript and useful comments, particularly on the non-trivial $SU(2)$ symmetry of the spin-1 XY model\,\cite{odea2024levelstatistics}. He also appreciates Sanjay Moudgalya for discussions on related problems.

\appendix
\section{Some aspects of symmetric group}
\label{app:SN}

We begin with two classical results. The first one is a Schur-Weyl-type decomposition of the tensor product of a finite dimensional vector space $V$,
\begin{equation}\label{eq:SW decomp}
    V^{\otimes N}\cong\bigoplus_T\mb{S}_TV\cong\bigoplus_\lambda \mb{S}_\lambda V\otimes V_\lambda,
\end{equation}
where $T$ runs over standard Young tableaux, $\mb{S}_TV$ (or $\mb{S}_\lambda V$) denotes the corresponding Weyl module. Passing from $T$ to $\lambda$ means that we forget the numbering and keep only the Young diagram: if $T$ and $T'$ share the same diagram, then $\mb{S}_TV\cong\mb{S}_{T'}V$. The rightmost decomposition makes explicit that $V^{\otimes N}$ splits into irreducible representations (irreps) $V_\lambda$ of the symmetric group and the corresponding multiplicity spaces $\mb{S}_\lambda V$ that turns out to be an irrep of $U(V)$. Note that one of these summands will realize $\mr{Sym}^N(\mf{h}^s)$ in our setting. 

The second result concerns tensor products of Weyl modules,
\begin{equation}
    \mb{S}_\lambda V\otimes\mb{S}_\mu V\cong\bigoplus_\nu N_{\lambda\mu\nu}\mb{S}_\nu V,
\end{equation}
where the non-negative integer $N_{\lambda\mu\nu}$ is known as the Littlewood-Richardson (LR) coefficient (a representation-theoretic analogue of Clebsch-Gordan coefficient).

\subsection{Irreducible representations}
The symmetric group $\mf{S}_N$ consists of permutations of the set $\{1,\cdots,N\}$, with group multiplication given by composition. Being finite, $\mf{S}_N$ admits only finite-dimensional representations, and any such representation decomposes into a direct sum of irreps. We recall the following standard result for a finite group $G$:
\begin{theorem}\label{thm:irreps}
For any finite group $G$, the regular representation $\mb{C}[G]$ contains every irrep as a subrepresentation, 
\begin{equation}\label{eq:regular rep}
    \mb{C}[G]\cong\bigoplus_\lambda V_\lambda^{\oplus\dim V_\lambda},
\end{equation}
where $\lambda$ labels irreps and $V_\lambda$ denotes the corresponding irrep.
\end{theorem}
Here we identify the regular representation with the group algebra $\mb{C}[G]$, viewed as a $G$-module via multiplication. 

For $G=\mf{S}_N$, irreps are conveniently labeled by Young diagrams. Concretely, a Young diagram $\lambda$ is determined by a partition $\lambda=(\lambda_1,\cdots,\lambda_M)$ of $N$ satisfying $\sum_{i=1}^M\lambda_i=N$ and $\lambda_i\geq\lambda_{i+1}$. The diagram consists of $N$ boxes arranged in $M$ left-aligned rows, with $\lambda_i$ boxes in the $i$th row (see the left panel of Fig.\,\ref{fig:tableau} for a diagram for $N=5$). Each such diagram labels an irrep of $\mf{S}_N$, and every irrep arises this way; in particular, the trivial representation corresponds to the one-row diagram $(N)$.

To construct the irrep explicitly, one chooses a standard Young tableau $T$, i.e., a numbering of the boxes by $\{1,\cdots,N\}$ that increase along each row and along each column (see the right-panel of Fig.\,\ref{fig:tableau}). We define two subgroups of $\mf{S}_N$ associated with $T$: $\mf{P}_T$ permutes entries within each row, and $\mf{Q}_T$ permutes entries within each column, 
\begin{equation}\label{eq:P Q}
    \begin{split}
        \mf{P}_T&\coloneqq\left\{\sigma\in\mf{S}_N;\,\sigma\text{ preserves each row}\right\}\\
        \mf{Q}_T&\coloneqq\left\{\sigma\in\mf{S}_N;\,\sigma\text{ preserves each column}\right\}.
    \end{split}
\end{equation}
From these we build the symmetrizer and the anti-symmetrizer,
\begin{equation}
    a_T\coloneqq\sum_{\sigma\in\mf{P}_T}\sigma,\,\,b_T\coloneqq\sum_{\sigma\in\mf{Q}_T}\mr{sgn}(\sigma)\sigma,
\end{equation}
and the Young symmetrizer as their product
\begin{equation}
    c_T\coloneqq a_Tb_T\big(\in\mb{C}[\mf{S}_N]\big),
\end{equation}
where $\mr{sgn}(\sigma)\in\{\pm1\}$ is the parity of $\sigma$. Since $\mb{C}[\mf{S}_N]$ is a module over itself, we may regard $a_T, b_T$, and $c_T$ as endomorphisms $\mc{L}(\mb{C}[\mf{S}_N])$ (in the following we denote endomorphisms by $\mc{L}$). The corresponding irrep is realized as the right ideal,
\begin{equation}
    V_T=\mb{C}[\mf{S}_N]c_T,
\end{equation}
where $c_T$ can be interpreted as a projector.

It is important to distinguish tableaux $T$ and $T'$ that share the same diagram $\lambda$. Let $c_T$ and $c_{T'}$ be the corresponding Young symmetrizers. One can show (see Appendix\,\ref{proof:ctct'=0}) that $c_Tc_{T'}=0$ for $T\not=T'$, implying that the corresponding right ideals $V_T(=\mb{C}[\mf{S}_N]c_T)$ and $V_{T'}(=\mb{C}[\mf{S}_N]c_{T'})$ are linearly independent. At the same time, they are canonically isomorphic (via relabeling of tableau entries). Together with Thm\,\ref{thm:irreps}, this yields the well-known relation\footnote{This consideration does not solely imply Eq.\,\eqref{eq:dim V} (only $\dim V_\lambda\geq\#\text{ of numberings of the diagram }\lambda$ can be concluded). Nevertheless Eq.\,\eqref{eq:dim V} turns out to be correct.},
\begin{equation}\label{eq:dim V}
    \dim V_\lambda=\#\text{ of numberings of the diagram }\lambda. 
\end{equation}

We also require a basic fact about inducing representations from $\mf{S}_N\times\mf{S}_M$ to $\mf{S}_{N+M}$ (often called the outer product). Given irreps $V_\lambda$ of $\mf{S}_N$ and $V_\mu$ of $\mf{S}_M$, we view $V_\lambda\boxtimes V_\mu$ as a representation of $\mf{S}_N\times\mf{S}_M$ and induce it to $\mf{S}_{N+M}$ via the natural inclusion $\mf{S}_N\times\mf{S}_M\hookrightarrow\mf{S}_{N+M}$. Denoting the induced representation by $V_\mu\circ V_\nu$, one obtains
\begin{equation}\label{eq:LR decomp}
    V_\lambda\circ V_\mu\cong\bigoplus_\nu N_{\lambda\mu\nu}V_\nu, 
\end{equation}
where $N_{\lambda\mu\nu}$ indicates the multiplicity of the representation and is called the Littlewood-Richardson (LR) coefficient. In our application, we will repeatedly use the special case
\begin{equation}
   N_{\lambda\mu(M+N)}\not=0\,\,\text{only if}\,\,\lambda=(N),\,\mu=(M).
\end{equation}
Intuitively, the trivial representation of $\mf{S}_{N+M}$ can arise from the outer product only when both constituents are trivial.

\subsection{Commutant and collective generators}

A standard bicommutant statement implies that $H$ decomposes into irreducible $G$-sectors together with irreducible sectors of $\mr{Comm}(G)$:
\begin{theorem}\label{thm:SW}
    The module $\mc{H}$ decomposes into irreps of $G$ as,
    \begin{equation}\label{eq:SW}
        \mc{H}\cong\bigoplus_\lambda \mc{H}_\lambda^{\oplus\dim V_\lambda},
    \end{equation}
    where $V_\lambda$ is an irrep of $G$ and $\mc{H}_\lambda$ is an irrep of $\mr{Comm}(G)$.
\end{theorem}
A proof can be found in standard references on finite-group representations\,\cite{Fulton2004}. The basic idea is to construct projectors $P_\lambda$ onto irreps using the group algebra (c.f. Thm\,\ref{thm:irreps}), and then define $\mc{H}_\lambda\coloneqq\mc{H}P_\lambda$. Note that each irreducible block is isomorphic to a tensor product
\begin{equation}\label{eq:decomp bicommutant}
    \mc{H}_\lambda^{\oplus\dim V_\lambda}\cong V_\lambda\otimes \mc{H}_\lambda.
\end{equation}
In the QMBS context, analogous decompositions arise when one replaces the group algebra by an appropriate bond algebra\,\cite{moudgalya2023}. Eq.\,\eqref{eq:SW} and Eq.\,\eqref{eq:decomp bicommutant} then correspond to splitting the Hilbert space into the scar subspace and its complement. 

We now apply Thm\,\ref{thm:SW} to the present problem by taking $\mc{H}=\bigotimes_{x\in\Lambda}\mf{h}_x^s$ and $G=\mf{S}_N$. On a product basis $\ket{e_1}\otimes\cdots\otimes\ket{e_N}$, the right action of $\sigma\in G$ is defined by
\begin{equation}
    \hat{\sigma}^R\ket{e_1}\otimes\cdots\otimes\ket{e_N}\coloneqq\ket*{e_{\sigma^{-1}(1)}}\otimes\cdots\otimes\ket*{e_{\sigma^{-1}(N)}},
\end{equation}
where the superscript $R$ emphasizes that this is a right action (and hence a group homomorphism). We also define the corresponding left action by
\begin{equation}
    \hat{\sigma}\ket{e_1}\otimes\cdots\otimes\ket{e_N}\coloneqq\ket*{e_{\sigma(1)}}\otimes\cdots\otimes\ket*{e_{\sigma(N)}}.
\end{equation}
These are related by $\hat{\sigma}^R=\hat{\sigma}^{-1}$. Importantly, the right action satisfies $\hat{\sigma}^R\hat{\tau}^R=\widehat{\sigma\tau}^R$, whereas the left action composes in the opposite order ($\what{\sigma\tau}=\hat{\tau}\hat{\sigma}$).

The same constructions extend to the group algebra $\mb{C}[\mf{S}_N]$. Since $\mb{C}[\mf{S}_N]$ decomposes via Young symmetrizers, Thm.\,\ref{thm:SW} implies an analogous decomposition of the many-body Hilbert space,
\begin{equation}\label{eq:decomp H young}
    \mc{H}=\bigoplus_T\big(\hat{c}_\lambda^R\mc{H}\big)\eqqcolon\bigoplus_T\mb{S}_TV,
\end{equation}
where $T$ runs over standard tableaux. As a basic example, the projector onto the totally symmetric subspace $\mr{Sym}^N(V)$ is constructed as 
\begin{equation}
    \hat{c}_{(N)}^R=\frac{1}{|\mf{S}_N|}\sum_{\sigma\in\mf{S}_N}\hat{\sigma}^R=\frac{1}{|\mf{S}_N|}\sum_{\sigma\in\mf{S}_N}\hat{\sigma},
\end{equation}
where $(N)$ denotes the one-row Young diagram.

We conclude by connecting the LR rule to Weyl modules. Eq.\,\eqref{eq:LR decomp} implies the corresponding relation for Young symmetrizers, $c_\lambda c_\mu=\sum_\nu N_{\lambda\mu\nu}c_\nu$, where $\lambda,\mu$ and $\nu$ are partitions of $N,M$, and $N+M$, respectively. Applying this to $V^{\otimes N}$ yields 
\begin{equation}\label{eq:LR coeff.}
    \mb{S}_\lambda V\otimes\mb{S}_\mu V\cong\bigoplus_\nu N_{\lambda\mu\nu}\mb{S}_\nu V,
\end{equation}
which is the module-level counterpart of Eq.\,\eqref{eq:LR decomp}.

\section{Proofs}
\subsection{Proof of the orthogonality in different tableaux}\label{proof:ctct'=0}
In this section we provide a self-contained proof of the orthogonality of the Young symmetrizer $c_Tc_{T'}=0$ for different (standard) tableaux. For a more detailed account one may refer to any introductory textbook of the representation theory (e.g., Ref.\,\cite{Fulton2004}).

As introduced in Sec.\,\ref{subsec:rep of SN}, the irreps of the symmetric group $\mf{S}_N$ are completely classified by the Young diagram $\lambda=(\lambda_1,\cdots,\lambda_k)$ with $\sum_i\lambda_i=N$ and $\lambda_1\geq\lambda_2\geq\cdots\geq\lambda_k\geq1$, with the corresponding Young symmetrizer defined as $c_T=a_Tb_T$\footnote{Notice that unlike the main text the definitions depend on a tableau $T$.}. The factors $a_T$ and $b_T$ are the symmetrizer and the anti-symmetrizer, 
\begin{equation}\begin{split}
    &a_T\coloneqq\sum_{\sigma\in\mf{P}_T}\sigma,\,\,b_T\coloneqq\sum_{\sigma\in\mf{Q}_T}\mr{sgn}(\sigma)\sigma,
\end{split}
\end{equation}
where $\mf{P}_T$ and $\mf{Q}_T$ are subgroups of $\mf{S}_N$,
\begin{equation}
    \begin{split}
        \mf{P}_T&\coloneqq\big\{\sigma\in\mf{S}_N; \sigma\text{ preserves each row}\big\}\\
        \mf{Q}_T&\coloneqq\big\{\sigma\in\mf{S}_N; \sigma\text{ preserves each column}\big\}.
    \end{split}
\end{equation}

We prove the orthogonality of the Young symmetrizers:
\begin{lemma}\label{lemma: young symmetrizer appendix}
    For the Young symmetrizer $c_T$, the following holds,
    \begin{itemize}
        \item[(i)] If $T\not=T'$, then $c_Tc_{T'}=0$
        \item[(ii)] For $\forall x\in\mb{C}[\mf{S}_N]$, $c_Txc_T$ is proportional to $c_T$. In particular, $c_Tc_T=n_Tc_T$ with $n_T=|\mf{S}_N|/\dim V_T$ with $V_T\coloneqq\mb{C}[\mf{S}_N]c_T$.  
    \end{itemize}
\end{lemma}
To show this lemma we first prove the following claim,
\begin{lemma}\label{lemma:transposition}
    If $\sigma\not\in\mf{P}_T\cdot\mf{Q}_T$ (here the group multiplication of the subgroups are defined in a trivial manner), then there exists at least one transposition $t=(i,j)\in\mf{P}_T$ such that $q=\sigma^{-1}t\sigma$ is an element of $\mf{Q}_T$.
\end{lemma}
\begin{proof}
    Let $T'=\sigma T$ be another tableau with the same diagram by replacing the numbers $1,2,\cdots$ in the tableau $T$ with their permutation by $\sigma(1),\sigma(2),\cdots$. If there exists a pair of integers $i,j$ such that $i$ and $j$ are in the same row in $T$, while they are in the same column in $T'$, then $t=(i,j)$ is trivially a desired element of $\mf{Q}_T$. We will show that if such a pair does not exist, $\sigma$ must be written in the form $g=p\cdot q$ for some $p\in\mf{P}_T$ and $q\in\mf{Q}_T$. 
    
    Since we assume that any pair of two integers that are in the same row in $T$ is not in the same column in $T'$. One can permute the numbers in the tableaux $T$ and $T'$ such that there exist $p_1\in\mf{P}_T$ and $q'_1\in \sigma\mf{Q}_T\sigma^{-1}$ such that $p_1T$ and $q'_1T'$ has the identical first row. We repeat this procedure to conclude that there exist $p\in\mf{P}_T$ and $q'\in\mf{Q}_T$ such that $pT=q'T'=q'\sigma T$. This in turn means $\sigma=pq$ by setting $q=\sigma^{-1}(q')^{-1}\sigma$. 
\end{proof}

\begin{lemma}\label{lemma:lemma for young symmetrizer}
    \begin{itemize}
        \item[(i)] If the permutation $p$ is an element of $\mf{P}_T$, i.e., $p\in\mf{P}_T$, then $p\cdot a_T=a_T\cdot p=a_T$
        \item[(ii)] If the permutation $q$ is an element of $\mf{Q}_T$, i.e., $q\in\mf{Q}_T$, then $(\mr{sgn}(q)q)\cdot b_T=b_T\cdot(\mr{sgn}(q)q)=q$
        \item[(iii)] For $\forall p\in\mf{P}_T$ and $\forall q\in\mf{Q}_T$, $p\cdot c_T\cdot(\mr{sgn}(q)q)=c_T$, and only $c_T$ satisfies this condition (as an element of $\mb{C}[\mf{S}_N]$).  
    \end{itemize}
\end{lemma}
\begin{proof}
    Only the uniqueness statement in (iii) is nontrivial. Suppose $\sum_\sigma\alpha_\sigma\sigma\in\mb{C}[\mf{S}_N]$ satisfies the same condition (i.e., $p\cdot\sum_\sigma\alpha_\sigma\sigma\cdot(\mr{sgn}(q))q=\sum_\sigma\alpha_\sigma\sigma$). This implies,
    \begin{equation}\label{eq:condition invariance}\begin{split}
        &p\cdot\sum_\sigma\alpha_\sigma\sigma\cdot(\mr{sgn}(q)q)=\sum_\sigma\mr{sgn}(q)\alpha_\sigma p\sigma q\overset{!}{=}\sum_\sigma\alpha_\sigma\sigma\\
        &\to\mr{sgn}(q)\alpha_\sigma=\alpha_{p\sigma q}.
    \end{split}
    \end{equation} 
    By setting $\sigma=1$ one obtains $\mr{sgn}(q)\alpha_1=\alpha_{pq}$. It suffices for the uniqueness to show that if $g\not\in\mf{P}_T\cdot\mf{Q}_T$ then $\alpha_\sigma=0$. 
    
    By Lemma\,\ref{lemma:transposition}, for $\sigma\not\in\mr{P}_T\cdot\mf{Q}_T$ we can find $t$ such that $q=\sigma^{-1}t\sigma$ thus $t\sigma q=g$. Therefore, Eq.\,\eqref{eq:condition invariance} indicates $\alpha_\sigma=\mr{sgn}(q)\alpha_{tq\sigma}=\mr{sgn}(q)\alpha_q=-\alpha_q$, implying $\alpha_\sigma=0$.
\end{proof}

\begin{proof}[Proof of Lemma\,\ref{lemma: young symmetrizer appendix}]

    (i) Let us define the lexicographical order for tableaux. 
    
    Suppose that tableaux $T$ and $T'$ are constructed from different Young diagrams  $\lambda=(\lambda_1,\cdots,\lambda_k)$ and $\mu=(\mu_1,\cdots,\mu_{k'})$. We the define the ordering as follows,
    \begin{equation}
        T>T'\overset{\mr{def}}{\Leftrightarrow}\lambda_i-\mu_i>0 \text{ for the first $i$ s.t., }\lambda_i\not=\mu_i.
    \end{equation}

    For tableaux $T$ and $T'$ with the same diagram (i.e., just the numbering is different), we define the ordering as follows: we read the numbers from the top row to the bottom row. Suppose in a certain row (say $i$th row) the numbers written $T$ and $T'$ differ for the first time. In other words, up to the $i-1$th row the numbers (and the diagram itself) of the tableaux $T$ and $T'$ are identical. Reading the numbers of the $i$th row from left to right, we must find the box in which one of the following is the numbers in $T$ and $T'$ differ for the first time. We denote these numbers in the box as $n$ for $T$ and $n'$ for $T'$. We then define $T>T'$ if $n-n'>0$. 
    
    To show $c_Tc_{T'}=0$ for $T<T'$, it suffices to show that for $b_T\cdot a_{T'}=0$. $T<T'$ implies that one can find a pair of integers $(i,j)$ such that they lie in the same column in $T$, while in the same row in $T'$. By Lemma\,\ref{lemma:lemma for young symmetrizer}, this implies that for $t=(i,j)$ the following relations hold,
    \begin{equation}
        b_T\cdot t=-b_T,\,\,t\cdot a_T=a_T,
    \end{equation} 
    which in turn implies,
    \begin{equation}
        b_T\cdot a_{T'}=b_T\cdot(t\cdot a_{T'})=(b_T\cdot t)a_{T'}=-b_T\cdot a_{T'}.
    \end{equation}

    For $T>T'$, we show that for $\forall x\in\mb{C}\mf{S}_N$ we have $a_T\cdot x\cdot b_{T'}=0$. For $x=\sigma\in\mf{S}_N$, the statement $a_T\cdot \sigma\cdot b_{T'}=0$ is equivalent to $a_T\cdot b_{\sigma T'}=0$. If $T>\sigma T'$, one can conclude $a_T\cdot b_{\sigma T'}=0$ by the same transposition trick above. If $T<\sigma T'$, we employ ``the anti-involution trick'': $\widetilde{}:\mb{C}[\mf{S}_N]\to\mb{C}\mf{S}_N$ by linearly extending $\sigma\mapsto\widetilde{\sigma}\coloneqq \sigma^{-1}$ for $\sigma\in\mf{S}_N$. We find
    \begin{equation}
        \widetilde{a_T\cdot b_{\sigma T'}}=\widetilde{b}_{\sigma T'}\cdot\widetilde{a}_T=b_{\sigma T'}\cdot a_T.
    \end{equation}
    Therefore, the transposition trick implies $\widetilde{a_T\cdot b_{\sigma T'}}=0$. Since the anti-involution is injective, we have $a_T\cdot b_{\sigma T'}=0$.

    (ii) By (i) and (ii) of Lemma\,\ref{lemma:lemma for young symmetrizer}, we have $p\cdot(c_T\cdot c_T)\cdot q=c_T\cdot c_T$ for $\forall p\in\mf{P}_T$ and $\forall q\in\mf{Q}_T$. By (iii) of Lemma\,\ref{lemma:lemma for young symmetrizer} $c_T\cdot x\cdot c_T$ must be proportional to $c_T$, i.e., $c_T\cdot c_T=n_T c_T$ with some number $n_T$. Now we define the map $F:\mb{C}\mf{S}_N\to\mb{C}\mf{S}_N$ as a right multiplication of $n_T$, i.e., $F(a)=a\cdot c_T$. Note that we can regard $F$ as a linear map or a matrix such that the trace is properly defined. From the discussion above, we conclude $F\big|_{V_T}=n_T$ and $F\big|_{V_{T'}}=0$ for $T\not=T'$. This implies that the trace of this operator is $\mr{tr}F=n_T\dim V_T$. Another method to calculate the trace is simply to check the map $F$ itself: since $c_T$ contains $1$ only once, we have,
    \begin{equation}
        F(a)=a+\sum_{\substack{\sigma\in\mf{P}_T\cdot\mf{Q}_T\\ \sigma\not=1}}a\cdot\sigma,
    \end{equation}
    implying $\mr{tr}F=|\mf{S}_N|$. Therefore we obtain $n_T=|\mf{S}_N|/\dim V_T$. 
\end{proof}

\subsection{Proof of Theorem\,\ref{thm:GL(V)}}\label{proof:GL(V)}
The commutant of $\mf{S}_N$ is a set of operators that commute with all the permutations,
\begin{equation}
    \mr{Comm}(\mf{S}_N)=\left\{\hat{W}\in\mc{L}(W^{\otimes N});\,\hat{\sigma}\hat{W}\hat{\sigma}^{-1}=\hat{W}\,\,\text{for }\forall\sigma\in\mf{S}_N\right\}.
\end{equation}
Since $\mc{L}(V)\cong V\otimes V^*$ (the isomorphism is given by the ``vectorization'' in the context of open quantum systems or operator hydrodynamics), the operation $\hat{\sigma}\bullet\hat{\sigma}^{-1}$ can be seen as the $\mf{S}_N$ action on the module $\mc{L}(V^{\otimes N})$. The commutant then implies that,
\begin{equation}
    \hat{W}\in\mr{Comm}(\mf{S}_N)\Leftrightarrow\hat{W}=\mc{S}(\hat{W}),
\end{equation}
where $\mc{S}$ is the symmetrization defined in Eq.\,\eqref{eq:symmetrization}. Regarding $\mc{L}(V^{\otimes N})\cong(V\otimes V^*)^{\otimes N}$ as a $\mf{S}_N$ module, $\mc{S}$ corresponds to $\hat{c}_{(N)}^R$. Therefore we have
\begin{equation}
    \hat{W}=\mc{S}(\hat{W})\Leftrightarrow\hat{W}\in\mr{Sym}^N(V\otimes V^*)=\mr{Sym}^N(\mc{L}(V))=\left\langle \hat{U}^{\otimes N};\,\hat{U}\in\mc{L}(V)\right\rangle.
\end{equation}
Since $GL(V)$ is dense in $\mc{L}(V)$, $W$ is spanned by the tensor product of operators in $GL(V)\cong\mb{C}\otimes U(V)$, implying
\begin{equation}
    \hat{W}\in\left\langle\hat{U}^{\otimes N};\,\hat{U}\in U(V)\right\rangle.
\end{equation}

\subsection{Proof of Lemma\,\ref{lemma:su(V)}}\label{proof:su(V)}
(1) We expand a weight-preserving operator $\hat{O}$ by the eigenstate of $\hat{E}^{ii}(=\dyad*{i})$, diagonal elements of $\mf{su}(\mf{h}^s)$, denoted by $\ket*{\vec{n}}$, i.e., $\hat{E}^{ii}\ket*{\vec{n}}=n^i\ket*{\vec{n}}$, 
    \begin{equation}
        \hat{O}=\sum_{\{n_x\}_x}\sum_{\{n'_x\}_x}\mel{\vec{n}_1,\cdots,\vec{n}_N}{\hat{O}}{\vec{n}'_1,\cdots,\vec{n}'_N}\dyad*{\vec{n}_1,\cdots,\vec{n}_N}{\vec{n}'_1,\cdots,\vec{n}'_N}.
    \end{equation}
    The condition on $\hat{O}$ implies $\sum_{x=1}^Nn^i_x=\sum_{x=1}^Nn'^i_x$ for $\forall i=1\sim d_s$ with $d_s=\dim\mf{h}^s$. Since $\vec{n}_x$ is one of the standard basis of $\mb{R}^{d_s}$, we can decompose each term in $\hat{O}$ as a combination of a projection and a permutation: let $\sigma\in\mf{S}_N$ be the permutation that satisfies,
    \begin{equation}
        \hat{\sigma}\ket*{\vec{n}_1,\cdots,\vec{n}_N}=\ket*{\vec{n}'_1,\cdots,\vec{n}'_N}.
    \end{equation}
    Then, the term $\dyad*{\vec{n}_1,\cdots,\vec{n}_N}{\vec{n}'_1,\cdots,\vec{n}'_N}$ is expressed as,
    \begin{equation}\label{eq:basis decomp}
        \dyad*{\vec{n}_1,\cdots,\vec{n}_N}{\vec{n}'_1,\cdots,\vec{n}'_N}=\dyad*{\vec{n}_1,\cdots,\vec{n}_N}\hat{\sigma}^{-1}.
    \end{equation}
    The first projection is trivially expressed as a product of $\hat{E}^{ii}_x$, hence a polynomial of the Cartan subalgebra. This completes the proof. 

    (2) We expand an arbitrary operator $\hat{O}$ by the eigenstate of $\hat{\sigma}^3$ denoted by $\ket*{\tau}\,(\tau=\pm1)$, i.e., $\hat{\sigma}^3\ket*{\tau}=\tau\ket*{\tau}$,
    \begin{equation}
        \hat{O}=\sum_{\{\tau_x\}_x}\sum_{\{\tau'_x\}_x}\mel*{\tau_1,\cdots,\tau_N}{\hat{O}}{\tau'_1,\cdots,\tau'_N}\dyad*{\tau_1,\cdots,\tau_N}{\tau'_1,\cdots,\tau'_N}.
    \end{equation}
    For a moment we assume $\sum_x\tau_x\geq\sum_x\tau'_x$, i.e., the sequence $\{\tau_x\}_{x=1}^N$ contains more $+1$'s than $\{\tau'_x\}_{x=1}^N$. There exists a permutation $\sigma\in\mf{S}_N$ such that $\tau_x\geq\tau'_{\sigma(x)}$ for any $x$. Then, the term $\dyad*{\tau_1,\cdots,\tau_N}{\tau'_1,\cdots,\tau'_N}$ is expressed as,
    \begin{equation}
        \dyad*{\tau_1,\cdots,\tau_N}{\tau'_1,\cdots,\tau'_N}=\dyad*{\tau_1,\cdots,\tau_N}{\tau'_{\sigma(1)},\cdots,\tau'_{\sigma(N)}}\hat{\sigma}.
    \end{equation}
    The first term is composed of $\hat{\sigma}_x^3$ if $\tau_x=\tau'_{\sigma(x)}$, and $\hat{\sigma}_x^+$ if $\tau_x\geq\tau'_{\sigma(x)}$.

\subsection{Proof of Lemma\,\ref{lemma:local HA}}\label{proof:local HA}
Due to Eq.\,\eqref{eq:LR decomp}, the totally symmetric weight basis $\ket*{\psi^\Lambda_{\vec{m}}}\in\mr{Sym}^N(\mf{h}^s)$ is decomposed as
    \begin{equation}
        \ket*{\psi^\Lambda_{\vec{m}}}=\sum_{\vec{m}',\vec{m}''}c_{\vec{m}}^{\vec{m}'\vec{m}''}\ket*{\psi^X_{\vec{m}'}}_X\otimes\ket*{\psi^{X^c}_{\vec{m}''}}_{X^c},
    \end{equation}
    where $X^c$ is the complement of $X^c$ and $\ket*{\psi^X_{\vec{m}'}}_X$ and $\ket*{\psi^{X^c}_{\vec{m}''}}_{X^c}$ are the totally symmetric weight basis defined in $\mr{Sym}^{|X|}(\mf{h}^s)$ and $\mr{Sym}^{N-|X|}(\mf{h}^s)$, respectively. Since $\{\ket*{\psi^{X^c}_{\vec{m}''}}_{X^c}\}_{\vec{m}''}$ is a set of linearly independent states, the condition $\hat{O}_X\ket*{\psi_{\vec{m}}^\Lambda}=0$ is equivalent to 
    \begin{equation}
        \hat{O}_X\ket*{\psi^X_{\vec{m}'}}_X=0\,\,\text{if }c^{\vec{m}',\vec{m}''}_{\vec{m}}\not=0.
    \end{equation}
    Considering all the weight basis $\{\ket*{\psi^\Lambda_{\vec{m}}}\}_{\vec{m}}$ yields the condition $\hat{O}_X\ket*{\psi^X_{\vec{m}'}}=0$ for $\forall\vec{m}'$. We can then apply Lemma\,\ref{lemma:main} to $\bigotimes_{x\in X}\mf{h}_x$, confirming the decomposition Eq.\,\eqref{eq:annihilator X decomp}.

\subsection{Proof of Lemma\,\ref{lemma:su(2) linearity}}\label{proof:su(2) linearity}
Acting the operator on the eigenstate of $\hat{S}^3$, the linear independence implies
\begin{equation}
    \sum_{a}\xi_{a,b}\big(\hat{S}^3\big)^a\big(\hat{J}^+\big)^b=0,\,\sum_c\chi_{c,d}\big(\hat{S}^3\big)^c\big(\hat{J}^-\big)^d=0.
\end{equation}
Acting these operators again on the eigenstate of $\hat{S}^3$, we obtain the condition,
\begin{equation}\begin{split}
    &\sum_{a=1}^{N_a}\xi_{a,b}m^a=0\,\text{ for }-j+b\leq m\leq j\\
    &\sum_{c=1}^{N_b}\chi_{c,d}m^c=0\,\text{ for }-j\leq m\leq j-c.
    \end{split}
\end{equation}
These equations can be viewed as linear equations, where the corresponding matrices have elements $m^a$ and $m^c$, which have the rank equal or higher than the dimension of the corresponding vectors $\vec{\xi}_b\coloneqq(\xi_{a,b})_a$ and $\vec{\chi}_d\coloneqq(\chi_{c,d})_c$. Therefore, the only solution is $\xi_{a,b}=\chi_{c,d}=0$.

\bibliographystyle{unsrt}
\bibliography{ref.bib}
\end{document}